\documentclass{article}
\usepackage{amsmath}
\usepackage{amsthm}
\usepackage{amsfonts}
\usepackage{amssymb}
\usepackage{enumerate}
\usepackage[pdftex]{graphicx}
  \graphicspath{{./pics/}}
  \DeclareGraphicsExtensions{.pdf}
\usepackage{tikz}
\usetikzlibrary{arrows,automata,shapes,snakes}
\usepackage{tabularx}

\usepackage{wrapfig}
\newtheorem{theorem}{Theorem}
{\bfseries}{\itshape}
{\bfseries}{\rm}
{\bfseries}{\itshape}
\newtheorem{example}[theorem]{Example}{\bfseries}{\itshape}
\newtheorem{lemma}[theorem]{Lemma}{\bfseries}{\itshape}
{\itshape}{\itshape}
{\itshape}{\itshape}
\newtheorem{claim}[theorem]{Claim}{\itshape}{\itshape}

\newenvironment{proofof}[1]
{\medskip\noindent \textbf{#1.} \itshape} 
{}

\newcommand{\nat}{\mathbb{N}}

\newcommand{\closure}{\text{closure}}
\newcommand{\layer}{\text{layer}}
\newcommand{\cell}{\text{cell}}

\newcommand{\union}{\text{u}}

\newcommand{\lcs}{\text{lcs}}
\newcommand{\rev}{\text{rev}}
\newcommand{\diff}[1]{\text{D}^#1}

\newcommand{\sufforder}{{\preceq_s}}

\newcommand{\eps}{\varepsilon}
\newcommand{\N}{\mathbb{N}}
\newcommand{\F}{\mathcal{F}}

\newcommand{\U}{\mathcal{U}}
\newcommand{\embeds}{\preceq}
\newcommand{\Al}{\mathtt{Alph}}

\newcommand{\trans}[1]{\xrightarrow{#1}}
\newcommand{\A}{\mathcal{A}}
\newcommand{\B}{\mathcal{B}}
\newcommand{\NFA}{\mathcal{A}}
\newcommand{\LA}{L^{\A}}
\newcommand{\LB}{L^{\B}}
\newcommand{\KA}{K^{\A}}
\newcommand{\KB}{K^{\B}}
\newcommand{\MA}{M^{\A}}
\newcommand{\MB}{M^{\B}}
\newcommand{\pref}{\mathtt{pref}}
\newcommand{\midd}{\mathtt{mid}}
\newcommand{\suff}{\mathtt{suff}}
\newcommand{\graph}{\texttt{SYNCH}}

\renewcommand{\epsilon}{\varepsilon}

\begin{document}
\title{Efficient Separability of Regular Languages\\ by Subsequences and
Suffixes}
\author{
  Wojciech Czerwi\'nski \qquad Wim Martens \qquad Tom\'{a}\v{s} Masopust
  \\
  \small{Institute for Computer Science, University of Bayreuth}
}

\maketitle
\begin{abstract}
  When can two regular word languages $K$ and $L$ be separated by a
  simple language? We investigate this question and consider
  separation by piecewise- and suffix-testable languages and variants
  thereof. We give characterizations of when two languages can be
  separated and present an overview of when these problems can be
  decided in polynomial time if $K$ and $L$ are given by
  nondeterministic automata.     \end{abstract}

\section{Introduction}
\makeatletter{}In this paper we are motivated by scenarios in which we want to
describe something complex by means of a simple language.  The
technical core of our scenarios consists of \emph{separation}
problems, which are usually of the following form:  \vspace{-2mm}
\begin{quote}
  Given are two languages $K$ and $L$. Does there exist a language
  $S$, coming from a family $\F$ of \emph{simple} languages, such that
  $S$ contains everything from $K$ and nothing from $L$?
\end{quote}
\vspace{-2mm}
The family $\F$ of simple languages could be, for example, languages
definable in FO, piecewise testable languages, or languages definable
with small automata.

Our work is specifically motivated by two seemingly orthogonal problems
coming from practice: (a) increasing the user-friendliness of XML Schema
and (b) efficient approximate query answering. We explain these next.

Our first motivation comes from simplifying XML Schema. XML Schema is
currently the only industrially accepted and widely supported schema
language for XML. Historically, it is designed to alleviate the
limited expressiveness of Document Type Definition (DTD) \cite{xml},
thereby making DTDs obsolete. Unfortunately, XML Schema's extra
expressiveness comes at the cost of simplicity. Its code is designed
to be machine-readable rather than human-readable and its logical
core, based on \emph{complex types}, does not seem well-understood by
users \cite{MartensNNS-vldb12}. One reason may be that the
specification of XML Schema's core \cite{xsd-1} consists of over 100
pages of intricate text. The BonXai schema language
\cite{MartensNNS-vldb12,MartensNNS-bonxai} is an attempt to overcome
these issues and to combine the simplicity of DTDs with the
expressiveness of XML Schema. It has exactly the same expressive power
as XML Schema, is designed to be human-readable, and avoids the use of
complex types. Therefore, it aims at simplifying the development or
analysis of XSDs. In its core, a BonXai schema is a set of rules $L_1
\to R_1, \ldots, L_n \to R_n$ in which all $L_i$ and $R_i$ are regular
expressions. An unranked tree $t$ (basically, an XML document) is in
the language of the schema if, for every node $u$, the word formed by
the labels of $u$'s children is in the language $R_k$, where $k$ is
the largest number such that the word of ancestors of $u$ is in
$L_k$. This semantical definition is designed to ensure full
back-and-forth compatibility with XML Schema \cite{MartensNNS-vldb12}.

When translating an XML Schema Definition (XSD) into an equivalent
BonXai schema, the regular expressions $L_i$ are obtained from a
finite automaton that is embedded in the XSD. Since the current
state-of-the-art in translating automata to expressions does not yet
generate sufficiently clean results for our purposes, we
are investigating simpler classes of expressions which we expect to
suffice in practice. Practical and theoretical studies show evidence
that regular expressions of the form $\Sigma^*w$ (with $w \in
\Sigma^+$) and $\Sigma^* a_1 \Sigma^* \cdots \Sigma^* a_n$ (with
$a_1,\ldots,a_n \in \Sigma$) and variations thereof seem to be quite
well-suited
\cite{GeladeN-jcss11,KasneciS-pods07,MartensNSB-tods06}. We
study these kinds of expressions in this paper.

Our second motivation comes from efficient approximate query
answering. Efficiently evaluating regular expressions is
relevant in a very wide array of fields. We choose one: in graph
databases and in the context of the SPARQL language
\cite{ArenasCP-www12,w3c-sparql11-2010,LosemannM-pods12,PerezAG-jws10}
for querying RDF data. Typically, regular expressions are used in this
context to match paths between nodes in a huge graph. In fact, the
data can be so huge that exact evaluation of a regular expression $r$
over the graph (which can lead to a product construction between an
automaton for the expression and the graph
\cite{LosemannM-pods12,PerezAG-jws10}) may not be feasible within
reasonable time. Therefore, as a compromise to exact evaluation, one
could imagine that we try to rewrite the regular expression $r$ as an
expression that we can evaluate much more efficiently and is close
enough to $r$. Concretely, we could specify two expressions
$r_\text{pos}$ (resp., $r_\text{neg}$) that define the language we want
to (resp., do not want to) match in our answer and ask whether there
exists a simple query (e.g., defining a piecewise testable
language) that satisfies these constraints. Notice that the scenario
of approximating an expression $r$ in this way is very general and not
even limited to databases. (Also, we can take $r_\text{neg}$ to be the
complement of $r_\text{pos}$.)

At first sight, these two motivating scenarios may seem to be
fundamentally different. In the first, we want to compute an
\emph{exact} simple description of a complex object and in the second
one we want to compute an \emph{approximate} simple query that can be
evaluated more efficiently. However, both scenarios boil down to the
same underlying question of language separation. Our contributions are:

\noindent (1) 
We formally define separation problems that closely correspond to the
  motivating scenarios. Query approximation will be abstracted as
  \emph{separation} and schema simplification as
  \emph{layer-separation} (Section \ref{sec:separability}).

\noindent (2)
We give a general characterization of separability of
  languages $K$ and $L$ in terms of boolean combinations of simple
  languages, layer-separability, and the existence of an infinite
  sequence of words that goes back and forth between $K$ and $L$. This
  characterization shows how the exact and approximate scenario are
  related and does not require $K$ and $L$ to be regular
  (Sec.~\ref{sec:characterization}). Our characterization generalizes
  a result by Stern \cite{Stern-tcs85} that says that a regular
  language $L$ is piecewise testable iff every increasing
  infinite sequence of words (w.r.t.\ subsequence ordering) alternates
  finitely many times between $L$ and its complement. 

  \noindent (3) In Section~\ref{sec:ptime} we prove a decomposition
  characterization for separability of regular languages by piecewise
  testable languages and we give an algorithm that decides
  separability. The decomposition characterization is in the spirit of
  an algebraic result by Almeida \cite{Almeida-jpaa90}.  It is
  possible to prove our characterization using Almeida's result but we
  provide a self-contained, elementary proof which can be understood
  without a background in algebra. We then use this characterization
  to distill a polynomial time decision procedure for separability of
  languages of NFAs (or regular expressions) by piecewise testable
  languages.       The state-of-the-art algorithm for separability by piecewise
  testable languages (\cite{Almeida-pmd99,AlmeidaZ-ita97}) runs in
  time $O(\text{poly}(|Q|)\cdot 2^{|\Sigma|})$ when given DFAs for the
  regular languages, where $|Q|$ is the number of states in the DFAs
  and $|\Sigma|$ is the alphabet size. Our algorithm runs in time
  $O(\text{poly}(|Q|+|\Sigma|))$ even for NFAs. We explain the
  connection to \cite{Almeida-pmd99,AlmeidaZ-ita97} more closely in the Appendix.       Notice that $|\Sigma|$ can be large (several hundreds and
  more) in the scenarios that motivate us, so we believe the
  improvement with respect to the alphabet to be relevant in practice.

\noindent (4)
  Whereas Section~\ref{sec:ptime} focuses exclusively on separation by
piecewise testable languages, we broaden our scope in
Section~\ref{sec:other}. Let's say that a \emph{subsequence language}
is a language of the form $\Sigma^* a_1 \Sigma^* \cdots \Sigma^* a_n
\Sigma^*$ (with all $a_i \in \Sigma$). Similarly, a \emph{suffix
  language} is of the form $\Sigma^*a_1 \cdots a_n$.  We present an
overview of the complexities of deciding whether regular languages can
be separated by subsequence languages, suffix languages, finite unions
thereof, or boolean combinations thereof. We prove all cases to be in
polynomial time, except separability by a single subsequence language
which is NP-complete. By combining this with the results from 
Section~\ref{sec:characterization} we also have that layer-separability
is in polynomial time for all languages we consider.

We now discuss further related work.  There is a large body
of related work that has not been mentioned yet. Piecewise testable
languages are defined and studied by Simon~\cite{Simon1972,Simon1975}, who
showed that a regular language is piecewise testable iff its syntactic
monoid is \textsf{J}-trivial and iff both the
minimal DFA for the language and the minimal DFA for the reversal are
partially ordered.
Stern~\cite{Stern85a} suggested an $O(n^5)$ algorithm in the size of a
DFA to decide whether a regular language is piecewise testable. This
was improved to quadratic time by
Trahtman~\cite{Trahtman2001}. (Actually, from our proof, it now
follows that this question can be decided in polynomial time if an NFA
and its complement NFA are given.)

Almeida~\cite{Almeida-pmd99} established a connection between a number
of separation problems and properties of families of monoids called
pseudovarieties.  Almeida shows, e.g., that deciding whether two given
regular languages can be separated by a language with its syntactic
monoid lying in pseudovariety \textsf{V} is algorithmically equivalent
to computing two-pointlike sets for a monoid in pseudovariety
\textsf{V}. It is then shown by Almeida et al.~\cite{AlmeidaCZ-jpaa08}
how to compute these two-pointlike sets in the pseudovariety
\textsf{J} corresponding to piecewise testable languages.
Henckell et al. \cite{DBLP:journals/ijac/HenckellRS10a} and Steinberg
\cite{Steinberg-SF01} show that the two-pointlike sets can be computed
for pseudovarieties corresponding to languages definable in first
order logic and languages of dot depth at most one, respectively. By
Almeida's result~\cite{Almeida-pmd99} this implies that the separation
problem is also decidable for these classes.

\section{Preliminaries and Definitions}
\makeatletter{}For a finite set $S$, we denote its cardinality by $|S|$.  By $\Sigma$
we always denote an alphabet, that is, a finite set of symbols.  A
($\Sigma$-)\emph{word} $w$ is a finite sequence of symbols $a_1\cdots
a_n$, where $n\ge 0$ and $a_i \in \Sigma$ for all $i = 1,\ldots,n$.
The \emph{length of $w$}, denoted by $|w|$, is $n$ and the
\emph{alphabet of $w$}, denoted by $\Al(w)$, is the set
$\{a_1,\ldots,a_n\}$ of symbols occurring in $w$.  The empty word is
denoted by $\epsilon$.  The set of all $\Sigma$-words is denoted by
$\Sigma^*$.  A \emph{language} is a set of words. For $v = a_1 \cdots a_n$ and $w \in \Sigma^* a_1 \Sigma ^* \cdots
\Sigma^* a_n \Sigma^*$, we say that $v$ is a \emph{subsequence} of
$w$, denoted by $v \preceq w$.

A {\em (nondeterministic) finite automaton} or {\em NFA} $\NFA$ is a
tuple $(Q,\Sigma,\delta,q_0,F) $, where $Q$ is a finite set of states,
$\delta : Q \times \Sigma \to 2^Q$ is the transition function, $q_0
\in Q$ is the initial state, and $F \subseteq Q$ is the set of
accepting states.  We sometimes denote that $q_2 \in \delta(q_1,a)$ as
$q_1 \trans{a} q_2 \in \delta$ to emphasize that $\NFA$ being in state
$q_1$ can go to state $q_2$ reading an $a \in \Sigma$.  A \emph{run of $\NFA$
  on word $w = a_1 \cdots a_n$} is a sequence of states $q_0 \cdots
q_n$ where, for each $i = 1,\ldots,n$, we have $q_{i-1} \trans{a_i}
q_i \in \delta$.  The run is \emph{accepting} if $q_n\in F$.  Word $w$
is \emph{accepted} by $\NFA$ if there is an accepting run of $\NFA$ on $w$.
The \emph{language of $\NFA$}, denoted by $L(\NFA)$, is the set of all words
accepted by $\NFA$.  By $\delta^*$ we denote the extension of $\delta$ to
words, that is, $\delta^*(q,w)$ is the set of states that can be
reached from $q$ by reading $w$.  The \emph{size} $|\NFA|=|Q| + \sum_{q,a}
|\delta(q,a)|$ of $\NFA$ is the total number of transitions and states.
An NFA is \emph{deterministic} (a \emph{DFA}) when every $\delta(q,a)$ consists
of at most one element.

The \emph{regular expressions (RE)} over $\Sigma$ are defined as
follows: $\varepsilon$ and every $\Sigma$-symbol is a regular
expression; whenever $r$ and $s$ are regular expressions, then so are
$(r\cdot s)$, $(r + s)$, and $(s)^*$.  In addition, we allow
$\emptyset$ as a regular expression, but we assume that $\emptyset$
does not occur in any other regular expression. For
readability, we usually omit concatenation operators and parentheses
in examples. We sometimes abbreviate an $n$-fold concatenation of
$r$ by $r^n$.  The \emph{language}
defined by an RE $r$ is denoted by $L(r)$ and is defined as
usual. Often we simply write $r$ instead of $L(r)$. Whenever we say that expressions or automata are
\emph{equivalent}, we mean that they define the same language.  The
\emph{size} $|r|$ of $r$ is the total number of occurrences of
alphabet symbols, epsilons, and operators in $r$, i.e., the number of
nodes in its parse tree.  A regular expression is \emph{union-free} if
it does not contain the operator $+$.  A language is \emph{union-free}
if it is defined by a union-free regular expression.

A quasi-order is a reflexive and transitive relation. 
For a quasi-order $\preccurlyeq$, the \emph{(upward) $\preccurlyeq$-closure} of a language
$L$ is the set $\closure^\preccurlyeq(L)=\{w \mid v \preccurlyeq w \text{ for some }
v\in L\}$. We denote the $\preccurlyeq$-closure of a word $w$ as
$\closure^\preccurlyeq(w)$ instead of $\closure^\preccurlyeq(\{w\})$.  Language $L$ is
\emph{(upward)  $\preccurlyeq$-closed} if $L=\closure^\preccurlyeq(L)$. 

A quasi-order
$\preccurlyeq$ on a set $X$ is a \emph{well-quasi-ordering} (a
\emph{WQO}) if for every infinite sequence $(x_i)_{i=1}^{\infty}$ of
elements of $X$ there exist indices $i < j$ such that $x_i\preccurlyeq
x_j$. It is known that every WQO is also \emph{well-founded}, that is,
there exist no infinite descending sequences $x_1 \succcurlyeq x_2 \succcurlyeq
\cdots$ such that $x_i \not \preccurlyeq x_{i+1}$ for all $i$.

Higman's Lemma \cite{higman} (which we use
multiple times) states that, for every alphabet $\Sigma$, the
subsequence relation $\preceq$ 
is a WQO on $\Sigma^*$. Notice that, as a corollary to
Higman's Lemma, every $\preceq$-closed language is a finite
union of languages of the form $\Sigma^* a_1 \Sigma^* \ldots \Sigma^*
a_n \Sigma^*$ which means that it is also regular, see
also~\cite{Ehrenfeucht1983}.  A language is \emph{piecewise testable}
if it is a finite boolean combination of $\preceq$-closed
languages (or, finite boolean combination of languages $\Sigma^* a_1
\Sigma^* \cdots \Sigma^* a_n \Sigma^*$). In this paper, all boolean
combinations are finite.

\subsection{Separability of Languages}\label{sec:separability}
A language $S$ \emph{separates language $K$ from $L$} if 
$S$ contains $K$ and does not intersect $L$. We say that $S$
\emph{separates $K$ and $L$} if it either separates $K$ from $L$ or
$L$ from $K$.
Let $\F$ be a family of languages.
Languages $K$ and $L$ are \emph{separable by $\F$} if
there exists a language $S$ in $\F$ that separates $K$ and $L$.
Languages $K$ and $L$ are \emph{layer-separable by $\F$}
if there exists a finite sequence of languages $S_1, \ldots, S_m$ in $\F$ such that
\begin{enumerate}
\item for all $1\le i \le m$, language $S_i \setminus
  \bigcup_{j=1}^{i-1} S_j$ intersects at most one of $K$ and $L$;
\item $K$ or $L$ (possibly both) is included in $\bigcup_{j=1}^{m}
  S_j$.
\end{enumerate}

Notice that separability always implies layer-separability. However,
the opposite implication does not hold, as we demonstrate next.

\begin{example}\label{ex:layer_separation}
 Let $\F = \{a^n a^* \mid n \geq 0\}$ be a family of
  $\preceq$-closed languages over $\Sigma = \{a\}$, 
    $K=\{a,a^3\}$, and $L=\{a^2,a^4\}$.  We first show that languages
  $K$ and $L$ are not separable by $\F$. Indeed, assume that $S \in
  \F$ separates $K$ and $L$.  If $K$ is included in $S$, then $aa^*
  \subseteq S$, hence $L$ and $S$ are not disjoint. Conversely, if $L
  \subseteq S$, then $a^2a^* \subseteq S$ and therefore $S$ and $K$ are
  not disjoint. This contradicts that $S$ separates $K$ and $L$.  Now
  we show that the languages are layer-separable by $\F$.  Consider
  languages $S_1= a^4a^*$, $S_2=a^3a^*$, $S_3=a^2a^*$, and $S_4=aa^*$.
  Then both $K$ and $L$ are included in $S_4$, and $S_1$ intersects
  only $L$, $S_2\setminus S_1=a^3$ intersects only $K$, $S_3\setminus
  (S_1\cup S_2)=a^2$ intersects only $L$, and $S_4\setminus (S_1\cup
  S_2\cup S_3)=a$ intersects only $K$; see Fig.~\ref{fig1}.  
\end{example} 

\begin{figure}
  \centering
  \begin{tikzpicture}[x=.7cm,y=.7cm]
        \draw (1,2) ellipse [x radius=.7cm, y radius=1.4cm];
    \draw (5,2) ellipse [x radius=.7cm, y radius=1.4cm];
    \draw (1,-.3) node {$K$};
    \draw (4.8,-.3) node {$L$};
    \draw (1,3) node {$a^3$};
    \draw (1,1) node {$a$};
    \draw (5,3) node {$a^4$};
    \draw (5,1) node {$a^2$};
            \draw (6.5,1.4) .. controls (3,1.2) and (2.6,4.7) .. (-.5,4.2);
    \draw (6.5,2.2) node {$S_1$};
        \draw (-.5,1.6) .. controls (3,1.6) and (3,1.6) .. (6.5,3.0);
    \draw (-.5,2) node {$S_2$};
    \draw (6.5,-.3) .. controls (3,-.5) and (2.6,3) .. (-.5,2.5);
    \draw (6.5, 0) node {$S_3$};
    \draw (-.5,-0.1) .. controls (3,-0.1) and (3,-0.1) .. (6.5,1.3);
    \draw (-.5,.2) node {$S_4$};
  \end{tikzpicture}
  \caption{An example of a layer-se\-pa\-ra\-tion.\label{fig1}}
\end{figure}
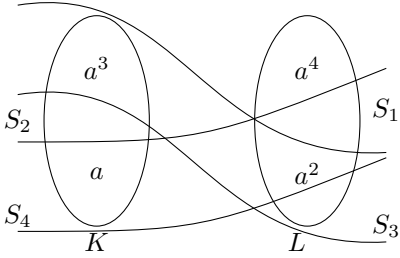
Example \ref{ex:layer_separation} illustrates some intuition behind
layered separability. Our motivation for layered separability comes
from the BonXai schema language which is discussed in the
introduction. We need to solve layer-separability if we want to decide
whether an XML Schema has an equivalent BonXai schema with simple regular
expressions (defining languages in $\F$). Layered separability implies that languages are, in a
sense, separable by languages from $\F$ in a priority-based system: If
we consider the ordered sequence of languages $S_1, S_2, S_3, S_4$
then, in order to classify a word $w \in K \cup L$ in either $K$ or
$L$, we have to match it against the $S_i$ in increasing order of the
index $i$. If we know the lowest index $j$ for which $w \in S_j$, we
know whether $w \in K$ or $w \in L$.

We now define a tool (similar to and slightly more general than the
\emph{alternating towers} of Stern~\cite{Stern-tcs85}) that allows us to
determine when languages are \emph{not} separable.
For languages $K$ and $L$ and a quasi-order $\preccurlyeq$, we say that a sequence $(w_i)_{i=1}^{k}$ of
words is a \emph{$\preccurlyeq$-zigzag between $K$ and $L$} if 
$w_1 \in K \cup L$ and, for all $i = 1, \ldots, k-1$:
\begin{center}
  (1) $w_i \preccurlyeq w_{i+1}$; (2) $w_i \in K$ implies $w_{i+1} \in L$; and
  (3) $w_i \in L$ implies $w_{i+1} \in K$.
\end{center}
We say that $k$ is the \emph{length} of the $\preccurlyeq$-zigzag.
We similarly define an infinite sequence of words to be an
\emph{infinite $\preccurlyeq$-zigzag between $K$ and $L$}.  If the languages $K$ and
$L$ are clear from the context then we sometimes omit them and refer to
the sequence as a \emph{(infinite) $\preccurlyeq$-zigzag}. 
If we consider the subsequence order $\preceq$, then we simply write a \emph{zigzag} instead of a $\preceq$-zigzag.
Notice that we do not require $K$ and $L$ to be disjoint. If there is a $w \in K \cap L$
then there clearly exists an infinite zigzag: $w, w, w, \ldots$

\begin{example}\label{ex:infinite_zigzag}
In order to illustrate infinite zigzags consider the
languages 
$K = \{a(ab)^{2k} c (ac)^{2\ell} \mid k, \ell \geq 0\}$
and
$L = \{b(ab)^{2k + 1} c (ac)^{2\ell + 1} \mid k, \ell \geq 0\}$.
Then the following infinite sequence is an infinite zigzag between $K$ and $L$:
\[
w_i =
\begin{cases}
  b (ab)^i c (ac)^i \hskip 1cm \text{if i is odd} \\
  a (ab)^i c (ac)^i \hskip 1cm \text{if i is even}
\end{cases}
\]
Indeed $w_1 \in L$, words from the sequence alternately belong to $K$ and $L$,
and for all $i \geq 1$ we have $w_i \preceq w_{i+1}$. \qed
\end{example}

\section{A Characterization of Separability}\label{sec:characterization}
\makeatletter{}The aim of this section is to prove the following theorem. It extends
a result by Stern that characterizes piecewise testable languages
\cite{Stern-tcs85}. In particular, it applies to general
languages and does not require $K$ to be the complement of $L$.
\begin{theorem}\label{theo:general-characterization}
  For languages $K$ and $L$ and a WQO $\preccurlyeq$ on words, the following are equivalent.
  \begin{enumerate}[(1)]
  \item $K$ and $L$ are separable by a boolean combination of $\preccurlyeq$-closed languages.
  \item $K$ and $L$ are layer-separable by $\preccurlyeq$-closed languages.
  \item There does not exist an infinite $\preccurlyeq$-zigzag between $K$ and $L$.
  \end{enumerate}
\end{theorem}
Some of the equivalences in the theorem still hold when the
assumptions are weakened. For example the equivalence between (1) and
(2) does not require $\preccurlyeq$ to be a WQO. 

Since the subsequence order $\preceq$ is a WQO on words, we know from
Theorem \ref{theo:general-characterization} that languages are
separable by piecewise testable languages if and only if they are
layer-separable by  $\preceq$-closed languages. Actually, since
$\preceq$ is a WQO (and therefore only has finitely many minimal
elements within a language), the latter is equivalent to being
layer-separable by languages of the form $\Sigma^* a_1 \Sigma^* \cdots
\Sigma^* a_n \Sigma^*$.

In Example~\ref{ex:layer_separation} we illustrated two languages
$K$ and $L$ that are layer-separable by  $\preceq$-closed
languages. Notice that $K$ and $L$ can also be separated by a boolean
combination of the languages $a^*a^1$, $a^*a^2$, $a^*a^3$, and $a^*a^4$ from $\F$, 
as $K \subseteq ((a^*a^1 \setminus a^*a^2) \cup (a^*a^3 \setminus
a^*a^4))$ and $L \cap ((a^*a^1 \setminus a^*a^2) \cup (a^*a^3 \setminus
a^*a^4)) = \emptyset$.

We now give an overview of the proof of
Theorem~\ref{theo:general-characterization}. The next lemma proves the
equivalence between (1) and (2), but is slightly more general. In
particular, it does not rely on a WQO.
\begin{lemma}\label{lem:theo1-2}
  Let $\F$ be a family of languages closed under intersection and
  containing $\Sigma^*$. Then languages $K$ and $L$ are separable by a
  finite boolean combination of languages from $\F$ if and only if $K$
  and $L$ are layer-separable by $\F$.
\end{lemma}
The proof (given in the Appendix) is constructive. The \emph{only if} direction
is the more complex one and shows how to exploit the implicit negation
in the first condition in the definition of layer-separability in
order to simulate separation by boolean combinations. Notice that the
families of
$\preccurlyeq$-closed languages in Theorem~\ref{theo:general-characterization}
always contain $\Sigma^*$ and are closed under intersection.

The following lemma shows that the implication (2) $\Rightarrow$ (3)
in Theorem~\ref{theo:general-characterization} does not require
well-quasi ordering.
\begin{lemma}\label{lem:layer_to_zigzag}
  Let $\preccurlyeq$ be a quasi order on words and assume
  that languages $K$ and $L$ are layer-separable by
   $\preccurlyeq$-closed languages.  Then there is no infinite
  $\preccurlyeq$-zigzag between $K$ and $L$.
\end{lemma}
To prove that (3) implies (2), we need the following technical lemma
in which we require $\preccurlyeq$ to be a WQO. In the proof of the
lemma, we argue how we can see $\preccurlyeq$-zigzags in a tree
structure. Intuitively, every path in the tree structure corresponds
to a $\preccurlyeq$-zigzag. We need the fact that $\preccurlyeq$ is a
WQO in order to show that we can assume that every node in this tree
structure has a finite number of children. We then apply K\"onig's
lemma to show that arbitrarily long $\preccurlyeq$-zigzags imply the existence of an
infinite $\preccurlyeq$-zigzag. The lemma then follows by contraposition. 
\begin{lemma}\label{lem:bound_for_zigzags}
  Let $\preccurlyeq$ be a WQO on words.  If there is no infinite
  $\preccurlyeq$-zigzag between languages $K$ and $L$, then
  there exists a constant $k \in \nat$ such that no
  $\preccurlyeq$-zigzag between $K$ and $L$ is longer than $k$.
\end{lemma}
If there is no infinite $\preccurlyeq$-zigzag, then we can put a bound
on the maximal length of zigzags by
Lemma~\ref{lem:bound_for_zigzags}. This bound actually has a close
correspondence to the number of ``layers'' we need to separate $K$ and $L$. 
\begin{lemma}\label{lem:zigzag_to_nonseparability}
  Let $\preccurlyeq$ be a WQO on words and assume that 
  there is no infinite $\preccurlyeq$-zigzag between languages $K$ and $L$.
  Then the languages $K$ and $L$ are layer-separable by  $\preccurlyeq$-closed languages.
\end{lemma}

\section{Testing Separability by Piecewise Testable Languages}\label{sec:testing}\label{sec:ptime}
\makeatletter{}Whereas Section 3 proves a result for general WQOs, we focus in this
section exclusively on the ordering $\preceq$ of
subsequences. Therefore, if we say \emph{zigzag} in this section, we always mean
\emph{$\preceq$-zigzag}. We show here how to decide the existence of
an infinite zigzag between two regular word languages, given by
their regular expressions or NFAs, in polynomial time. 
According to
Theorem~\ref{theo:general-characterization}, this is equivalent to
deciding if the two languages can be separated by a piecewise testable
language.

To this end, we first prove a decomposition result that is reminiscent
of a result of Almeida (\cite{Almeida-jpaa90}, Theorem 4.1
in~\cite{AlmeidaCZ-jpaa08}). We show that, if there is an infinite
zigzag between regular languages, then there is an infinite zigzag of
a special form and in which every word can be decomposed in some
synchronized manner. We can find these special forms of zigzags in
polynomial time in the NFAs for the languages. The main features are
that our algorithm runs exponentially faster in the alphabet size
than the current state-of-the-art \cite{AlmeidaZ-ita97} and that our algorithm and its proof
of correctness do not require knowledge of the algebraic perspective
on regular languages.

A regular language is a \emph{cycle language} if it is of the form $u
(v)^* w$, where $u, v, w$ are words and $(\Al(u) \cup \Al(w))
\subseteq \Al(v)$. We say that $v$ is the \emph{cycle} of the language
and that $\Al(v)$ is its \emph{cycle alphabet}.
Regular languages $\LA$ and $\LB$ are \emph{synchronized in one step} if they
are of one of the following forms:
\begin{itemize}
  \item $\LA = \LB = \{w\}$, that is, they are the same singleton word, or
  \item $\LA$ and $\LB$ are cycle languages with equal cycle alphabets.
\end{itemize}
We say that regular languages $\LA$ and $\LB$ are \emph{synchronized}
if they are of the form $\LA = D^{\A}_1 D^{\A}_2 \ldots D^{\A}_k$ and
$\LB = D^{\B}_1 D^{\B}_2 \ldots D^{\B}_k$
where, for all $1\le i\le k$, languages $D^{\A}_i$ and $D^{\B}_i$ are
synchronized in one step.  So, languages are synchronized if they can
be decomposed into (equally many) components that can be synchronized
in one step. Notice that synchronized languages are always non-empty.

\begin{example}\label{ex:synchronization}
  Languages $\LA=a (ba)^* aab \, ca \, bb (bc)^*$ and $\LB=b (aab)^*
  ba \, ca \, cc (cbc)^* b$ are synchronized.  Indeed, $\LA=D^{\A}_1
  D^{\A}_2 D^{\A}_3$ and $\LB=D^{\B}_1 D^{\B}_2 D^{\B}_3$ for
  $D^{\A}_1 = a (ba)^* aab$, $D^{\A}_2 = ca$, $D^{\A}_3 = bb (cb)^*$
  and $D^{\B}_1 = b (aab)^* ba$, $D^{\B}_2 = ca$, and $D^{\B}_3 = cc
  (cbc)^* b$.
\end{example}

The next lemma shows that, in order to search for infinite zigzags, it
suffices to search for synchronized sublanguages. The proof goes
through a sequence of lemmas that gradually shows how the sublanguages
of $\LA$ and $\LB$ can be made more and more specific.
\begin{lemma}[Synchronization / Decomposition]\label{lem:synchronization}
  There is an infinite zigzag between regular languages $\LA$ and
  $\LB$ if and only if there exist synchronized languages $\KA
  \subseteq \LA$ and $\KB \subseteq \LB$.
\end{lemma}
We now use this result to obtain a polynomial-time algorithm solving
our problem. The first step is to define what it means for NFAs
to contain synchronized sublanguages.

For an NFA $\A$ over an alphabet $\Sigma$, two states $p$, $q$, and
a word $w \in \Sigma^*$, we write $p \trans{w} q$ if $q \in
\delta^*(p,w)$ or, in other words, the automaton can go from state $p$
to state $q$ by reading $w$.  For $\Sigma_0\subseteq\Sigma$, states
$p$ and $q$ are \emph{$\Sigma_0$-connected} in $\A$ if there exists a
word $uvw \in \Sigma_0^*$ such that:
\begin{enumerate}
\item $\Al(v) = \Sigma_0$ and
\item there is a state $m$ such that $p \trans{u} m$, $m \trans{v} m$,
  and $m \trans{w} q$.
\end{enumerate}

Consider two NFAs $\A = (Q^{\A}, \Sigma, \delta^{\A}, q_0^{\A}, F^{\A})$ and 
$\B = (Q^{\B}, \Sigma, \delta^{\B}, q_0^{\B}, F^{\B})$.
Let $(q^\A, q^\B)$ and $(\bar{q}^\A, \bar{q}^\B)$ be in $Q^{\A} \times Q^{\B}$.
We say that $(q^\A, q^\B)$ and $(\bar{q}^\A, \bar{q}^\B)$ are
\emph{synchronizable in one step}\label{def:ndfsynch} if one of the
following situations occurs:
\begin{itemize}
\item there exists a symbol $a$ in $\Sigma$ such that $q^\A \trans{a}
  \bar{q}^\A$ and $q^\B \trans{a} \bar{q}^\B$,
\item there exists an alphabet $\Sigma_0 \subseteq \Sigma$ such that
  $q^{\A}$ and $\bar{q}^{\A}$ are $\Sigma_0$-connected in $\A$ and
  $q^{\B}$ and $\bar{q}^{\B}$ are $\Sigma_0$-connected in $\B$.
\end{itemize}
We say that automata $\A$ and $\B$ are \emph{synchronizable} if there
exists a sequence of pairs $(q^{\A}_0, q^{\B}_0), \ldots, (q^{\A}_k, q^{\B}_k) \in Q^{\A} \times Q^{\B}$ such that:
\begin{enumerate}
\item for all $0 \leq i < k$, $(q^{\A}_i, q^{\B}_i)$ and
  $(q^{\A}_{i+1}, q^{\B}_{i+1})$ are synchronizable in one step;
\item states $q^\A_0$ and $q^\B_0$ are initial states of $\A$ and
  $\B$, respectively; and
\item states $q^\A_k$ and $q^\B_k$ are accepting states of $\A$ and $\B$,
  respectively.
\end{enumerate}

Notice that if the automata $\A$ and $\B$ are synchronizable, then the
languages $L(\A)$ and $L(\B)$ are not necessarily synchronized, only
some of its sublanguages are necessarily synchronized.
\begin{lemma}[Synchronizability of automata]\label{lem:synchronizability}
  For two NFAs $\A$ and $\B$, the following conditions are equivalent.
  \begin{enumerate}
  \item Automata $\A$ and $\B$ are synchronizable.
  \item There exist synchronized languages $\KA \subseteq L(\A)$ and $\KB \subseteq L(\B)$.
  \end{enumerate}
\end{lemma}
The intuition behind Lemma~\ref{lem:synchronizability} is depicted in
Figure~\ref{fig:synch}.  The idea is that there is a sequence
$(q_0^\A,q_0^\B), \ldots, (q_k^\A,q_k^\B)$ that witnesses that $\A$
and $\B$ are synchronizable. The pairs of paths that have the same
style of lines depict parts of the automaton that are synchronizable in one step. In
particular, the dotted path from $q_1^\A$ to $q_j^\A$ has the same
word as the one from $q_1^\B$ to $q_j^\B$. The other two paths contain
at least one loop.

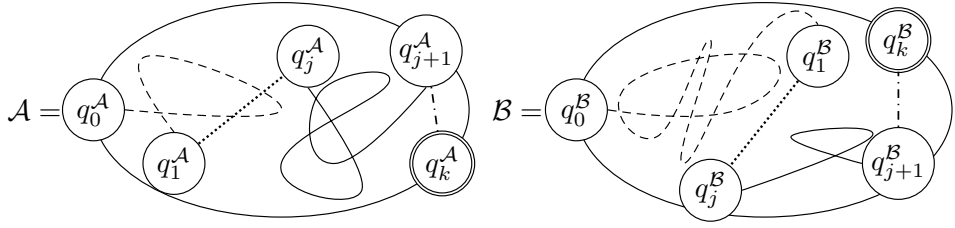
\begin{figure}[t]
  \centering
  \resizebox{\linewidth}{!}{
  \begin{tikzpicture}[scale = 0.65, inner sep = 0.3mm]
    \tikzstyle{vertex} = [circle, draw, fill = white, minimum size=7.5mm]
			  
    \draw (3.5, 2) ellipse [x radius = 3.5cm, y radius = 2cm];
    \draw (12.5, 2) ellipse [x radius = 3.5cm, y radius = 2cm];
    
    \node at (-1.1, 2) {$\A =$};
    \node at (7.9, 2) {$\B =$};
    
    \node[vertex] (a1) at (0, 2) {$q^\A_0$};
    \node[vertex] (a2) at (1.5, 1) {$q^\A_1$};
    \node[vertex] (a3) at (4, 3) {$q^\A_j$};
    \node[vertex] (a4) at (6.2, 3.1) {$q^\A_{j+1}$};
    \node[vertex, double] (a5) at (6.5, 1) {$q^\A_k$};
    
    \node[vertex] (b1) at (9, 2) {$q^\B_0$};
    \node[vertex] (b2) at (13.5, 3) {$q^\B_1$};
    \node[vertex] (b3) at (11.5, 0.5) {$q^\B_j$};
    \node[vertex] (b4) at (15, 1) {$q^\B_{j+1}$};
    \node[vertex, double] (b5) at (15, 3.3) {$q^\B_k$};    

    \draw [densely dashed] plot [smooth, tension = 1] coordinates {(a1.east) (3.5, 2) (1, 3) (a2.north)};
    \draw [densely dotted, thick] (a2) -- (a3);
    \draw plot [smooth, tension = 1] coordinates {(a3.south) (5, 0.5) (3.5, 1) (5.5, 2.5) (4.2, 2.3) (4.5, 1) (a4.south)};
        \draw [line width=.6pt,dash pattern=on .5pt off 2pt on 3pt off 2pt] (a4) -- (a5);

    \draw [densely dashed] plot [smooth, tension = 1]
      coordinates {(b1.east) (12, 2) (12.5, 3) (10, 2.5) (10.5, 1.5) (11.5, 3.3) (11, 1) (12.5, 3.5) (b2.north)};
    \draw [densely dotted, thick] (b2) -- (b3);
    \draw plot [smooth, tension = 1] coordinates {(b3.east) (14.5, 1.5) (13, 1.5) (b4.west)};
    \draw [line width=.6pt,dash pattern=on .5pt off 2pt on 3pt off 2pt] (b4) -- (b5);    
  \end{tikzpicture}}
  \caption{Synchronization of automata $\A$ and $\B$.}
  \label{fig:synch}
\end{figure}

The following theorem states that synchronizability in automata
captures exactly the existence of infinite zigzags between their
languages. The theorem statement uses Theorem~\ref{theo:general-characterization} for the
connection between infinite zigzags and separability.
\begin{theorem}\label{thm:characterization}
  Let $\A$ and $\B$ be two NFAs.  Then the languages $L(\A)$ and
  $L(\B)$ are separable by a piecewise testable language if and only if
  the automata $\A$ and $\B$ are not synchronizable.
\end{theorem}

We can now show how the algorithm from \cite{AlmeidaZ-ita97} can be
improved to test in polynomial time whether two given NFAs are
synchronizable or not. Our algorithm computes quadruples of states
that are synchronizable in one step and by linking such quadruples
together so that they form a pair of paths as illustrated in
Figure~\ref{fig:synch}.
\begin{theorem}\label{lem:algorithm}
  Given two NFAs $\A$ and $\B$, it is possible to test in polynomial
  time whether $L(\A)$ and $L(\B)$ can be separated by a piecewise
  testable language. 
\end{theorem}

\section{Asymmetric Separation and Suffix Order}\label{sec:other}\label{sec:others}
\makeatletter{}We present a bigger picture on efficient separations that are
relevant to the scenarios that motivate us. For example, we consider
what happens when we restrict the allowed boolean combinations of
languages. Technically, this means that separation is no longer symmetric.
Orthogonally, we also consider the suffix order
$\sufforder$ between strings in which $v \, \sufforder \, w$ if and only if
$v$ is a (not necessarily strict) suffix of $w$. An important technical difference with the
rest of the paper is that the suffix order is not a WQO. Indeed, the
suffix order $\sufforder$ has an infinite antichain, e.g., $a, ab,
abb, abbb, \ldots$ The results we present here for suffix order hold
true for prefix order as well.

Let $\F$ be a family of languages.  Language $K$ is \emph{separable
  from a language $L$ by $\F$} if there exists a language $S$ in $\F$
that separates $K$ from $L$, i.e., contains $K$ and does not intersect $L$.  Thus, if $L$ is closed
under complement, then $K$ is separable
from $L$ implies $L$ is separable from $K$.
The \emph{separation problem by $\F$} asks, given an
NFA for $K$ and an NFA for $L$, whether $K$ is separable from $L$ by
$\F$.

We consider separation by families of languages $\F(O,C)$, where $O$
(``order'') specifies the ordering relation and $C$ (``combinations'')
specifies how we are allowed to combine (upward) $O$-closed
languages. Concretely, $O$ is either the subsequence order $\preceq$
or the suffix order $\sufforder$. We allow $C$ to be one of
\emph{single}, \emph{unions}, or \emph{bc} (boolean combinations),
meaning that each language in $\F(O,C)$ is either the $O$-closure of a
single word, a finite union of the $O$-closures of single words, or a
finite boolean combination of the $O$-closures of single words. Thus,
$\F(\preceq,\textit{bc})$ is the family of piecewise testable
languages and $\F(\sufforder,\textit{bc})$ is the family of
suffix-testable languages.  With this convention in mind, the main
result of this section is to provide a complete complexity overview of
the six possible cases of separation by $\F(O,C)$. The case
$\F(\preceq,\text{bc})$ has been proved in Section~\ref{sec:ptime} and
the remaining ones are proved in the Appendix.
\begin{theorem}\label{theo:other}
  For $O \in \{\preceq,\sufforder\}$ and $C$ being one of single,
  unions, or boolean combinations, we have that the complexity of the
  separation problem by $\F(O,C)$ is as indicated in
  Table~\ref{tab:overview}.
\end{theorem}
Since the separation problem for prefix order is basically the same
as the separation for suffix order and has the same complexity we
didn't list it separately in the table. Furthermore, from
Lemma~\ref{lem:theo1-2} we immediately obtain that deciding
layer-separability for all six cases in Table~\ref{tab:overview} is in PTIME.

\begin{table}[t]
  \centering
  \begin{tabular}{|@{\hspace{3mm}}c@{\hspace{3mm}}||@{\hspace{5mm}}c@{\hspace{5mm}}|@{\hspace{5mm}}c@{\hspace{5mm}}|@{\hspace{2mm}}c@{\hspace{2mm}}|}
    \hline
    $\F(O,C)$ & single & unions & bc (boolean combinations) \\
    \hline \hline
    $\preceq$ (subsequence) & NP-complete & PTIME & PTIME \\
    \hline
    $\sufforder$ (suffix) & PTIME & PTIME & PTIME \\
    \hline
  \end{tabular}
  \caption{The complexity of deciding separability for
    regular languages $K$ and $L$.}
  \label{tab:overview}
\end{table}

\section{Conclusions and Further Questions}
\makeatletter{}Subsequence- and suffix languages seem to be very promising for
obtaining ``simple'' separations of regular languages, since we can
often efficiently decide if two given regular languages are separable
(Table~\ref{tab:overview}). Layer-separability is even in PTIME in all
cases. Looking back at our motivating scenarios,
the obvious next questions are: if a separation exists, can we
efficiently compute one? How large is it? 

If we look at the broader picture, we wonder if our characterization
of separability can be used in a wider context than regular languages
and subsequence ordering. Are there other cases where it can be used
lead to obtain efficient decision procedures? Another concrete
question is whether we can decide in polynomial time if a given NFA
defines a piecewise-testable language.  Furthermore, we are also
interested in efficient separation results by combinations of languages of the
form $\Sigma^* w_1 \Sigma^* \cdots \Sigma^* w_n$ or variants thereof.

\paragraph{Acknowledgments.}
We thank Jean-Eric Pin and Marc Zeitoun for patiently answering our
questions about the algebraic perspective on this problem. We are
grateful to Miko{\l}aj Boja{\'n}czyk, who pointed out the connection between layered
separability and boolean combinations. We also thank Piotr Hofman for
pleasant and insightful discussions about our proofs during his visit
to Bayreuth.

\appendix
\section*{Appendix}
\makeatletter{}\section*{Connection to the Algorithm of Almeida and Zeitoun}
Almeida and Zeitoun~\cite{AlmeidaZ-ita97} show that the following problem is in polynomial time:
\begin{description}
\item[Input:] Two DFAs $\A = (Q^\A, \Sigma, \delta^\A, q_0^\A, F^\A)$
  and $\B = (Q^\B, \Sigma, \delta^\B, q_0^\B, F^\B)$ with
  constant-size alphabet $\Sigma$.
\item[Problem:] Are $L(\A)$ and $L(\B)$ separable by a piecewise
  testable language?
\end{description}
A result by Almeida \cite{Almeida-pmd99} says that separability is
equivalent to computing the intersection of topological closures of
the regular languages that are to be separated. This is used by
Almeida and Zeitoun \cite{AlmeidaZ-ita97}, who prove that these
topological closures can be represented by a class of automata (going
beyond DFAs or NFAs) computable from the original automata. The main
differences with the present procedure are that the construction of
\cite{AlmeidaZ-ita97} is
\begin{enumerate}[(1)]
\item exponential in the
  size of the alphabet and
\item defined on DFAs
  rather than on NFAs.
\end{enumerate}
Actually, the exponential time bound w.r.t. the alphabet size has
already been observed by Almeida and Zeitoun in the conclusions of
their paper \cite{AlmeidaCZ-jpaa08}.
The reason why the algorithm from \cite{AlmeidaZ-ita97} is exponential
in the size of the alphabet is that its first step consists of adding,
to each loop that uses a subset $B$ of the alphabet $\Sigma$, a new
loop containing $B^{\omega}$. (See Definition 4.1 from
\cite{AlmeidaZ-ita97} -- the notation $B^{\omega}$ is borrowed from
that paper.) The number of these subsets can be
exponential. In fact, the algorithm from \cite{AlmeidaZ-ita97} first
adds these cycles to $\A$ and $\B$ separately and then (after some
more operations) compares the automata to each other. However, the
relevant cycles to add to $\A$ depend on $\B$.
\begin{example}
  Consider a language
  \[
  L(\A) = (a_1^* \cdots a_n^*)^*.
  \]
  Then, for every subset $S = \{i_1, \ldots, i_j\} \subseteq \{1, \ldots, n\}$,
  there exists a language
  \[
  L(\B_S) = (a_{i_1} \cdots a_{i_j})^*
  \]
  such that the intersection of the closures of the above languages contains $u^\omega$
  only for words $u$ such that $\Al(u) = \{a_{i_1}, \ldots, a_{i_j}\}$.
\end{example}
The example shows that if we compute the closure of $L(\A)$ without
looking simultaneously at $\B$ we have to keep all of the
exponentially many loops in order to be prepared for intersecting this
closure with the closure of any possible language $L(\B_S)$. In fact,
one needs to do more than na\"ively compute largest common
subsets of alphabets of loops that obviously correspond to each
other. We show how to do this while avoiding the exponent in $|\Sigma|$.

The following is a slightly less trivial example that shows how
alphabets of strongly connected components can correspond to each other.
\begin{example}
Consider languages
\[
(a^* b c^* d e^*)^* a c b^* a c (b a)^* c a (b c)^* b
\]
and
\[
(a b)^* d (a b^* d f^*)^* b (c (a b)^* c^* b^* a (c b)^*)^* b.
\]
These languages cannot be separated by a piecewise testable language.
The example of profinite word in the intersection of the closures is
\[
(a b d)^\omega c a c (a b)^\omega c a (b c)^\omega.
\]
(Again, the notation $m^\omega$ is borrowed from~\cite{AlmeidaZ-ita97} and is the standard
one for the unique idempotent power of element $m$ of the semigroup.)
\end{example}

\makeatletter{}\section*{Proofs of Section~\ref{sec:characterization}}

\begin{proofof}{Lemma~\ref{lem:theo1-2}}
  Let $\F$ be a family of
  languages closed under intersection and containing $\Sigma^*$. Then
  languages $K$ and $L$ are separable by a finite boolean combination
  of languages from $\F$ if and only if $K$ and $L$ are
  layer-separable by $\F$.
\end{proofof}

\begin{proof}
  For a sequence $S_1, S_2, \ldots, S_k$ denote, for all $i = 1, \ldots, k$
  \[
  \diff{S}_i = S_i \setminus \bigcup_{j=1}^{i-1} S_j\,.
  \]

  To show the {\it if} part, assume that $S_1,S_2, \ldots, S_m$ is the
  sequence of languages from $\F$ layer-separating $K$ and $L$. We
  will construct a finite boolean combination of languages from $\F$
  that separates $K$ and $L$.
  By definition of layer separability, each language $\diff{S}_i$ intersects at
  most one of $K$ and $L$. Furthermore, $K$ or $L$ is included in
  $\bigcup_{j=1}^{m} \diff{S}_j=\bigcup_{j=1}^{m} S_j$.  Without loss of
  generality, assume that $K$ is included in $\bigcup_{j=1}^{m} \diff{S}_j$,
  and set $J = \{j \mid \diff{S}_j \cap K \neq \emptyset\}$.  Then the
  language $S=\bigcup_{j \in J} \diff{S}_j$ separates $K$ and $L$.  As $S$ is
  a finite boolean combination of languages from $\F$, this part is
  shown.

  To show the {\it only if} part, assume that $K$ and $L$ can be
  separated by a language $S$ that is a finite boolean combination of
  languages from $\U = \{U_1, \ldots, U_k\}$, a finite subset of $\F$.  Without loss
  of generality, assume that $K \subseteq S$ and $L \cap S =
  \emptyset$.  For any subset of indices $I \subseteq \{1, \ldots,
  k\}$ we denote
  \[
  \cell_{\U}(I) = \Big( \bigcap_{i \in I} U_i \Big) \cap \Big( \bigcap_{i \notin I} \overline{U_i} \Big),
  \]
  where $\overline{U_i}=\Sigma^*\setminus U_i$ and call this language
  a \emph{cell}; see Fig.~\ref{fig:cells} for an
  illustration.
  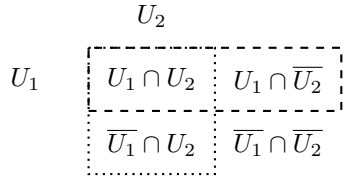
\begin{figure}[b]
    \centering
    \begin{tikzpicture}[x=1.4cm,y=.7cm]
            \draw [thick, dashed] (1.2, 1.2) -- (3.6, 1.2) -- (3.6, 2.4) -- (1.2, 2.4) -- (1.2, 1.2);
      \draw [thick, dotted] (1.2, 2.4) -- (1.2, 0) -- (2.4, 0) -- (2.4, 2.4) -- (1.2, 2.4);
            \draw (0.6, 1.8) node {$U_1$};
      \draw (1.8, 3) node {$U_2$};
                  \draw (1.8, 0.6) node   {$\overline{U_1}\cap U_2$};
      \draw (1.8, 1.8) node  {$U_1\cap U_2$};
      \draw (3, 0.6) node  {$\overline{U_1}\cap \overline{U_2}$};
      \draw (3, 1.8) node {$U_1\cap \overline{U_2}$};
    \end{tikzpicture}
    \caption{Cells for two languages $U_1$ and $U_2$.\label{fig:cells}}
  \end{figure}
  Observe that the cells are pairwise disjoint and $S$ is a finite
  union of cells.  As $S$ separates $K$ and $L$, every cell intersects
  at most one of the languages $K$ and $L$. The cells that form $S$ do
  not intersect $L$ and the others do not intersect $K$. Based on
  this, we construct a layer-separation of $K$ and $L$ by $\F$.

  To this end, we show that there exists a sequence of languages
  $S_1, \ldots, S_{2^k}$ from $\F$ and a bijection $\pi: \{1, \ldots, 2^k\} \to \mathcal{P}(\{1, \ldots, k\})$
  such that, for every $1 \leq j \leq 2^k$
  \[
  \diff{S}_j = \cell_{\U}(\pi(j)).
  \]

  We call $S_1, \ldots, S_{2^k}$ a sequence of \emph{cell-separating}
  languages for $\U$.  It is easy to see that this sequence $S_1,
  \ldots, S_{2^k}$ would layer-separate $K$ and $L$. 
  Indeed, for each $1 \leq \ell \leq 2^k$, the set $\diff{S}_\ell$ is a cell.
  Thus it intersects at most one of $K$ and $L$, which is the first requirement
  of a layer-separation. Moreover, the union
  $\bigcup_{1 \leq i \leq 2^k} S_i = \bigcup_{1 \leq i \leq 2^k} \diff{S}_i$ includes all the cells,
  so it equals $\Sigma^*$, thus clearly includes both $K$ and $L$.
 
  Therefore, it only remains to prove that there exists a sequence
  $S_1,\ldots,S_{2^k}$ of cell-separating languages for $\U$.
  Before we show it formally we present an illustrating example in order to
  give an intuition how the required sequence is constructed.
  For $\U = \{U_1, U_2, U_3\}$ the cell-separating sequence is as follows:
  \[
  U_1 \cap U_2 \cap U_3, U_2 \cap U_3, U_1 \cap U_3, U_3, U_1 \cap U_2, U_2, U_1, \Sigma^*
  \]
  We prove the fact in general by induction on $k$ that there is a sequence of
  cell-separating languages for every $k$-element set $\U \subseteq
  \F$.  For the base step, i.e., $k = 1$, we have that $\U =
  \{S_1\}$. We can simply take $S_1 = U_1$ and $S_2 = \Sigma^*$ and we
  are done.  Assume now that, for some $k$, the induction hypothesis
  is satisfied.  We prove it for $k+1$. Consider an arbitrary subset
  $\U' = \{U_1, \ldots, U_k, U_{k+1}\}$ of $\F$ and take $\U = \{U_1,
  \ldots, U_k\}$.  Let $S_1, \ldots, S_{2^k}$ be the sequence of
  cell-separating languages for $\U$.
  We will show that the sequence
  \[
  S_1 \cap U_{k+1}, \ldots, S_{2^k} \cap U_{k+1}, \, S_1, \ldots, S_{2^k}
  \]
  is cell-separating for $\U'$.
  We name this sequence $T_1, \ldots, T_{2^{k+1}}$, i.e.,
  \[
  T_i =
  \begin{cases}
    S_i \cap U_{k+1}, & \mbox{ if \ } i \leq 2^k \\
    S_{i-2^k}, & \mbox{ if \ } i > 2^k.
  \end{cases}
  \]
  
  It is sufficient to show that there exists a bijection $g$ between $\{1, \ldots, 2^{k+1}\}$ and $\mathcal{P}(\{1, \ldots, k+1\})$
  such that for $1 \leq i \leq 2^{k+1}$
  \[
  \diff{T}_i = \cell_{\U'}(\sigma(i)).
  \]
  Assume that $\pi$ is a bijection between $\{1, \ldots, 2^k\}$ and $\mathcal{P}(\{1, \ldots, k\})$
  such that
  \[
  \diff{S}_i = \cell_{\U}(\pi(i)).
  \]
  We will show that $\sigma$ defined as
  \[
  \sigma(i) =
  \begin{cases}
    \pi(i) \cup \{k+1\}, & \mbox{ if \ } i \leq 2^k \\
    \pi(i - 2^k), & \mbox{ if \ } i > 2^k
  \end{cases}
  \]
  fulfills the necessary condition.
  If $i \leq 2^k$ then
  \begin{align*}
    \diff{T}_i & = T_i \setminus \bigcup_{j=1}^{i-1} T_j = (S_i \cap U_{k+1})  \setminus \bigcup_{j=1}^{i-1} (S_j \cap U_{k+1})
      = (S_i \setminus \bigcup_{j=1}^{i-1} S_j) \cap U_{k+1} \\
    & = \diff{S}_i \cap U_{k+1} = \cell_{\U}(\pi(i)) \cap U_{k+1} = \cell_{\U'}(\sigma(i)).
  \end{align*}
  
  On the other hand if $i > 2^k$ then
  \begin{align*}
    \diff{T}_i & = T_i \setminus \bigcup_{j=1}^{i-1} T_j = (S_{i-2^k})  \setminus
    \Big( \bigcup_{j=1}^{2^k} (S_j \cap U_{k+1}) \cup \bigcup_{j=1}^{i-2^k-1} S_j \Big) \\
    & = (S_{i-2^k})  \setminus \Big(U_{k+1} \cup \bigcup_{j=1}^{i-2^k-1} S_j \Big)
    = \diff{S}_{i-2^k} \setminus U_{k+1} \\
    & = \cell_{\U}(\pi(i-2^k)) \setminus U_{k+1} = \cell_{\U'}(\sigma(i)),
  \end{align*}
  since $\bigcup_{j=1}^{2^k} S_j=\Sigma^*$,
  which completes the proof.
\end{proof}

\begin{proofof}{Lemma~\ref{lem:layer_to_zigzag}}
  Let $\preccurlyeq$ be a quasi-order on words and assume
  that languages $K$ and $L$ are layer-separable by
   $\preccurlyeq$-closed languages.  Then there is no infinite
  $\preccurlyeq$-zigzag between $K$ and $L$.
\end{proofof}

\begin{proof}
  For the sake of contradiction, assume that there exists an infinite
  $\preccurlyeq$-zigzag $(w_i)_{i=1}^{\infty}$ between $K$ and $L$.
  Let $I=\{w_1, w_2, \ldots\}$ and consider the sequence of languages
  $S_1, \ldots, S_m$ layer-separating $K$ and $L$.  Let $k$ in
  $\{1,\ldots,m\}$ be the lowest index for which $S_k\cap
  I\neq\emptyset$. (Notice that $k$ exists by definition of
  layer-separations.)
  Since we
  chose $k$ to be minimal, for every $j \geq 1$ it holds that
  \[
  w_j \notin \bigcup_{i=1}^{k-1} S_i\,.
  \]
  Let $\ell \geq 1$ be such that $w_{\ell} \in S_k\cap I$.
  Without loss of generality, assume that $w_{\ell} \in
  K$. (Otherwise, we switch $K$ and $L$.)  Then, by
  the definition of zigzag, $w_{\ell+1} \in L$.  As $S_k$ is 
  $\preccurlyeq$-closed and as $w_\ell \preccurlyeq w_{\ell+1}$, we have that also $w_{\ell+1} \in S_k$.  Thus,
  \begin{align*}
    w_{\ell+1} \in L \, \cap \, (S_k \setminus \bigcup_{i=1}^{k-1} S_i)
    && \text{ and } &&
    w_{\ell} \in K \, \cap \, (S_k \setminus \bigcup_{i=1}^{k-1} S_i)\,.
  \end{align*}
  But then the set $S_k \setminus \bigcup_{i=1}^{k-1} S_i$ intersects
  both languages $K$ and $L$, which is a contradiction with the
  assumption that $S_1, \ldots, S_m$ layer-separates $K$ and $L$.
\end{proof}

In the next proof we use K\"{o}nig's Lemma, which we recall next.  A
tree is \emph{finitely branching} if every node has finitely many
children. Note that, for every $n > 0$ there can be a node that has at least $n$ children.

\begin{lemma}[K\"{o}nig~\cite{konig}]\label{lem:konig}
A finitely branching tree containing arbitrarily long paths contains an infinite path.
\end{lemma}

\begin{proofof}{Lemma~\ref{lem:bound_for_zigzags}}
  Let $\preccurlyeq$ be a WQO on words.  If there is no infinite
  $\preccurlyeq$-zigzag between languages $K$ and $L$, then
  there exists a constant $k \in \nat$ such that no
  $\preccurlyeq$-zigzag between $K$ and $L$ is longer than $k$.
\end{proofof}

\begin{proof}
  All $\preccurlyeq$-zigzags considered in this proof are between languages $K$ and $L$.
  We show that the existence of arbitrarily long $\preccurlyeq$-zigzags implies the existence of an infinite $\preccurlyeq$-zigzag.
  To this end, we restrict the general form of $\preccurlyeq$-zigzag to be able to use K\"{o}nig's Lemma.
  Note that any WQO allows equivalent elements.  For a word $w$, let
  $[w]=\{ v \in \Sigma^* \mid v\preccurlyeq w \text{ and }
  w\preccurlyeq v \}$ denote the equivalence class containing $w$.
  For languages $K$ and $L$, we arbitrarily pick two elements from the
  sets $[w]\cap K$ and $[w]\cap L$, if they exist, denoted by $[w]_K$
  and $[w]_L$, respectively, and call them {\em canonical elements} of
  the class $[w]$.
  We say that a $\preccurlyeq$-zigzag $(w_i)_{i=1}^k$ is \emph{canonical} if it consists only of canonical elements, that is, 
  if $w_i\in K$ then $w_i=[w_i]_K$, and if $w_i\in L$ then $w_i=[w_i]_L$.
  Observe that if there exists a $\preccurlyeq$-zigzag of length $k$ then there also
  exists a canonical $\preccurlyeq$-zigzag of length $k$.
  Indeed, replacing all elements of the $\preccurlyeq$-zigzag with their corresponding canonical elements results in a canonical $\preccurlyeq$-zigzag.
  Thus, in what follows, we consider only canonical $\preccurlyeq$-zigzags. Note that we reduced the quasi order to an order.
  We say that a $\preccurlyeq$-zigzag $(w_i)_{i=1}^k$ is \emph{denser} than a $\preccurlyeq$-zigzag $(v_i)_{i=1}^k$ if
  \begin{itemize}
    \item $w_i \preccurlyeq v_i$, for all $1 \leq i \leq k$;
    \item $w_i \in K \iff v_i \in K$, for all $1 \leq i \leq k$, and also symmetrically for $L$; and
    \item there exists $1\le j\le k$ such that $w_j\neq v_j$.
  \end{itemize}
  A $\preccurlyeq$-zigzag is \emph{densest} if there is no denser $\preccurlyeq$-zigzag.

  Note that if a $\preccurlyeq$-zigzag $(w_i)_{i=1}^k$ is densest then $(w_i)_{i=1}^j$
  is also densest for any $j < k$. Indeed, if $(v_i)_{i=1}^j$ is denser than $(w_i)_{i=1}^j$ then
  $v_1, \ldots, v_j, w_{j+1}, \ldots, w_k$ is also a valid $\preccurlyeq$-zigzag, which is denser
  than $(w_i)_{i=1}^k$.
  Furthermore, observe that if there exists a $\preccurlyeq$-zigzag of length $k$ then there also
  exists a densest $\preccurlyeq$-zigzag of length $k$ because the denser order
  is well founded, as a suborder of a $k$-componentwise product of well founded orders.
  Thus, by the assumtions, there exist arbitrarily long densest $\preccurlyeq$-zigzags.
  Their first element belongs either to $K$ or to $L$. Without loss of generality, we may assume that
  there are arbitrarily long densest $\preccurlyeq$-zigzags starting in $K$.
  Note that the first word in every densest $\preccurlyeq$-zigzag is the shortest canonical element with respect to the order $\preccurlyeq$ among the canonical elements of $K$. 
  As the order $\preccurlyeq$ is a WQO there are only finitely many shortest canonical elements,
  thus there exists a word $w \in K$ such that there are arbitrarily long densest $\preccurlyeq$-zigzags starting from $w$.
  Consider a tree consisting of all these $\preccurlyeq$-zigzags forming its paths.
  By definition, this tree has arbitrary long paths. It is also finitely branching;
  otherwise, if a node has infinitely many children labelled by different words $v_1, v_2, \ldots$, the WQO property implies that
  we can find a pair of indices $i < j$ such that $v_i \preccurlyeq v_j$.
  Then the $\preccurlyeq$-zigzag obtained by choosing the path going through $v_j$ is not densest
  as we can change $v_j$ into $v_i$ in this zigzag obtaining the denser one.
  Thus, by Lemma~\ref{lem:konig}, this tree contains an infinite path that forms
  an infinite $\preccurlyeq$-zigzag.
\end{proof}

\begin{proofof}{Lemma~\ref{lem:zigzag_to_nonseparability}}
  Let $\preccurlyeq$ be a WQO on words and assume that 
  there is no infinite $\preccurlyeq$-zigzag between languages $K$ and $L$.
  Then the languages $K$ and $L$ are layer-separable by $\preccurlyeq$-closed languages.
\end{proofof}

\begin{proof}

  For two languages $X$ and $Y$, let
  \[
  \layer(X, Y)=\{w \in X \mid \text{there does not exist } w' \text{ in } Y \text{ such that } w \preccurlyeq w'\}
  \]
  denote the set of all words of $X$ that are not smaller or equal to
  a string of $Y$ in the WQO. 
              We first show the following claim:
  \begin{claim}\label{claim:one_layer}
    There exists a $\preccurlyeq$-closed language $S_{(X,Y)}$ such that $S_{(X,Y)} \cap Y =
    \emptyset$ and $S_{(X,Y)} \cap X = \layer(X, Y)$.
  \end{claim}
  The proof of the claim is simple. Let $S_{(X,Y)} = \bigcup_{w \in \layer(X,
    Y)} \closure^\preccurlyeq(w)$. By definition, $S_{(X,Y)}$ is  $\preccurlyeq$-closed.  For each
  $w$ in $\layer(X, Y)$, we have that $\closure^\preccurlyeq(w) \cap Y = \emptyset$
  by definition of $\layer(X,Y)$. Therefore, $S_{(X,Y)}
  \cap Y = \emptyset$.  Moreover, we have that $\layer(X, Y) = S_{(X,Y)} \cap
  X$ because $w \in \layer(X,Y)$ implies that
  $(\closure^\preccurlyeq(w) \cap X) \subseteq \layer(X,Y)$. This concludes the proof of the claim.

  We now proceed with the proof of
  Lemma~\ref{lem:zigzag_to_nonseparability}. Let $B$ be a constant
  such that no $\preccurlyeq$-zigzag between $K$ and $L$ is longer
  than $B$. This constant exists by Lemma~\ref{lem:bound_for_zigzags},
  since there is no infinite $\preccurlyeq$-zigzag between $K$ and
  $L$.  Define the languages $K_0 = K$, $L_0 = L$, and, for each $i
  \in \nat$,
  \begin{align*}
    K_{i+1} = K_i \setminus \layer(K_i, L_i)
    &&\text{ }&&
    L_{i+1} = L_i \setminus \layer(L_i, K_i)\,.
  \end{align*}
  We prove by induction on $i$ that every $\preccurlyeq$-zigzag between
  $K_i$ and $L_i$ has length at most $B - i$.  The claim holds for
  $K_0$ and $L_0$, thus consider $K_{i+1}$ and $L_{i+1}$, for $i\ge
  0$.
        Since $K_{i+1} \subseteq K_i$ and $L_{i+1} \subseteq L_i$ we have
  that every $\preccurlyeq$-zigzag between $K_{i+1}$ and $L_{i+1}$
  would also be a $\preccurlyeq$-zigzag between $K_i$ and $L_i$. By
  induction we know that every $\preccurlyeq$-zigzag between $K_i$ and
  $L_i$ has length at most $B-i$. Therefore, every
  $\preccurlyeq$-zigzag between $K_{i+1}$ and $L_{i+1}$ also has
  length at most $B-i$. It remains to prove that there cannot be
  a $\preccurlyeq$-zigzag of length $B - i$ between $K_{i+1}$ and
  $L_{i+1}$. For the sake of contradiction, assume
  that $(w_k)_{k=1}^{B-i}$ is a $\preccurlyeq$-zigzag between
  $K_{i+1}$ and $L_{i+1}$ of length $B-i$. We either have that
  $w_{B-i} \in K_{i+1}$ or $w_{B-i} \in L_{i+1}$. We prove the case
  $w_{B-i} \in K_{i+1}$ since the other case is analogous. Here, we
  have that $w_{B-i} \notin \layer(K_i, L_i)$. By definition of
  $\layer(K_i, L_i)$, there exists a $w \in L_i$ such that $w_{B-i}
  \preccurlyeq w$. But this means that the sequence $w_1, \ldots,
  w_{B-i}, w$ would be a $\preccurlyeq$-zigzag between $K_i$ and $L_i$ of
  length $B-i+1$, which is a contradiction.

  We therefore have that, for every $i$, every $\preccurlyeq$-zigzag
  between $K_i$ and $L_i$ has length at most $B - i$. In particular,
  this means that if $i \geq B$, every $\preccurlyeq$-zigzag between
  languages $K_i$ and $L_i$ has length at most zero. Since any word $w
  \in K_i \, \cup \, L_i$ would already be a $\preccurlyeq$-zigzag of
  length one, this means that $K_i = L_i = \emptyset$.  

  We now show how the languages $K$ and $L$ can be layer-separated by
   $\preccurlyeq$-closed languages.  Denote by $S_{(X,Y)}$ the
   $\preccurlyeq$-closed language obtained when applying
  Claim~\ref{claim:one_layer} to languages $X$ and $Y$.  Then, the sequence
  \begin{center}
    $S_{(K_0, L_0)}$, $S_{(L_0, K_0)}$, $S_{(K_1, L_1)}$, $S_{(L_1,
      K_1)}$, \ldots, $S_{(K_{B-1}, L_{B-1})}$, $S_{(L_{B-1}, K_{B-1})}$
  \end{center}
  covering the layers 
  \begin{center}
    $\layer(K_0, L_0)$, $\layer(L_0, K_0)$,
    \ldots,     $\layer(K_{B-1}, L_{B-1})$, $\layer(L_{B-1}, K_{B-1})$,
  \end{center}
  respectively, layer-separates $K$ and $L$.
        Condition 1 of the definition of layered separability is satisfied
  because all the languages covering layers with smaller numbers
  appear earlier in the sequence.  Condition 2 is true because the
  union of all the considered layers includes $K \cup L$.  
\end{proof}

\section*{Proofs of Section~\ref{sec:testing}}
\makeatletter{}\subsection*{Proof of Lemma~\ref{lem:synchronization} with Running Example}\label{subsec:synchronization_proof}

In this section we prove Lemma~\ref{lem:synchronization}.

\begin{proofof}{Lemma~\ref{lem:synchronization}}
  There is an infinite zigzag between regular languages $\LA$ and $\LB$
  if and only if
  there exist synchronized languages $\KA \subseteq \LA$ and $\KB \subseteq \LB$.
\end{proofof}

\medskip
To prove it, we need several auxiliary results showing that if there is an infinite zigzag between two regular languages, then there is also an infinite zigzag between their sublanguages of a special form.

To illustrate the proofs of this section we use a running example with regular languages
\[
  \LA_1 = a(b^* a)^* a (bb)^* ab ca bb (bc)^* + (ab^*c)^* + b^* c (cb)^*
\]  
and
\[
  \LB_1 = abd + b(aab)^* ba ca (b(cb^*)^* c)^* cc (cbc)^* b + (aa)^* + ba (bb)^*
\]
having an infinite zigzag between them. After each step we present how the considered languages
have been modified.

We say that language $K$ \emph{embeds} into $L$, denoted by $K \preceq
L$, if for every $v \in K$, there exists a $w \in L$ such that $v \preceq w$. 
In order to be consistent we also say here that word $v$ \emph{embeds} into word $w$
if $\{v\}$ embeds in $\{w\}$, i.e., $v$ is a subsequence of $w$.
Languages $\KA$ and $\KB$ are \emph{mutually embeddable}
if $\KA \embeds \KB$ and $\KB \embeds \KA$.  Note that there always
exists an infinite zigzag between nonempty mutually-embeddable
languages.

\begin{lemma}[Mutual embeddability]\label{lem:embedded}
  If there is an infinite zigzag between regular languages $\LA$ and $\LB$, then there exist nonempty mutually-embeddable
  regular languages $\KA \subseteq \LA$ and $\KB \subseteq \LB$.
\end{lemma}

\begin{proof}
  We define languages $\KA$ and $\KB$ and show that they possess the required properties.
  Let $I$ denote the set of all words that belong to any infinite zigzag between languages $\LA$ and $\LB$,
  and let $\KA = \LA\cap I$ and $\KB=\LB\cap I$.
  Then, for any $w \in \KA$, let $I_w$ denote an infinite zigzag containing $w$.
  As $I_w$ is infinite, there exists $w'\in I_w\cap\LB$ such that $w \embeds w'$, hence $w' \in \KB$.
  Therefore $\KA \embeds \KB$. The case $\KB \embeds \KA$ is analogous.
   
  It remains to show that $\KA$ and $\KB$ are regular. We prove it for $\KA$ since the case for $\KB$ is analogous.
  Let $M$ denote the set of all minimal words of $\LA \setminus \KA$,
  that is, words $w \in \LA \setminus \KA$ such that there is no $w' \in \LA \setminus \KA$ with $w' \embeds w$ and $w' \neq w$.
  Note that any distinct words $w$ and $w'$ of $M$ are incomparable, i.e., $w \not\embeds w'$ and $w' \not\embeds w$.
  By Higman's lemma,  $M$ is finite.
  If $w \in \LA \setminus \KA$, that is, $w \notin I$, then any $w' \in \LA$ with $w \embeds w'$
  also belongs to $\LA \setminus \KA$; otherwise, $w' \in \KA = \LA \cap I$ implies that $w \in I$, which is a contradiction. 
  Thus,
  \[
    \LA \setminus \KA = \LA \cap \bigcup_{w \in M} \closure(w)\,.
  \]
  Notice that the language $\bigcup_{w \in M} \closure(w)$ is $\preceq$-closed,
  hence regular,   and so is the language $\KA = \LA \setminus (\LA \setminus \KA)$.
\end{proof}

By Lemma~\ref{lem:embedded} our running example could be reduced to
\[
  \LA_2 = a(b^* a)^* a (bb)^* ab ca bb (bc)^* +  b^* c (cb)^*
\]  
and
\[
  \LB_2 = b(aab)^* ba ca (b(cb^*)^* c)^* cc (cbc)^* b + (aa)^* + ba (bb)^*,
\]
since words from $(ab^*c)^* \subseteq \LA_1$ and $abd \subseteq \LB_1$ does not belong
to any infinite zigzag.

Now we strengthen the result by imposing a union-free decomposition on the languages $\KA$ and $\KB$. 

\begin{lemma}[Union-free languages]\label{lem:union_free_embedded}
  If there is an infinite zigzag between regular languages $\LA$ and $\LB$, 
  then there exist nonempty mutually-embeddable union-free regular languages $\KA \subseteq \LA$ and $\KB \subseteq \LB$.
\end{lemma}

\begin{proof}
  By Lemma~\ref{lem:embedded}, there exist nonempty mutually-embeddable regular languages $\MA \subseteq \LA$ and $\MB \subseteq \LB$.
  By~\cite{nagy}, see also \cite{Afonin09}, every regular language can be expressed as a finite union of union-free languages, hence we have
  \begin{align*}
    \MA = D^{\A}_1 \cup D^{\A}_2 \cup \ldots \cup D^{\A}_k
    &&\text{ and } &&
    \MB = D^{\B}_1 \cup D^{\B}_2 \cup \ldots \cup D^{\B}_{\ell}\,,
  \end{align*}
  where all languages $D^{\A}_i$ and $D^{\B}_j$ are union-free.
  It remains to show that there exist $1 \leq i \leq k$ and $1 \leq j \leq \ell$
  such that $D^{\A}_i$ and $D^{\B}_j$ are mutually embeddable. 
  We first show that for each $D^{\A}_i$ there exists a $D^{\B}_j$ such that $D^{\A}_i \embeds D^{\B}_j$.
  Consider a union-free regular expression for $D^{\A}_i$.
  For any $n \in \N$, we define $w_n$ as a word obtained from the expression for $D^{\A}_i$ by replacing stars with $n$.
  For the expression $(a b^* c)^* (c b)^*$ we have $w_n = (a b^n c)^n (a b)^n$.
  Note that for every $w\in D^{\A}_i$, there exists $n \in \N$ and a word $w_n \in D^{\A}_i$ such that $w \embeds w_n$.
  Number $n$ can be chosen s $n = |w|$ since $w_n$ is in $D^{\A}_i$ by definition.
  
  Consider now the sequence $(w_n)_{n=1}^{\infty}$. 
  Every word $w_n$ can be embedded to a word of $\MB$, therefore to a word of $D^{\B}_j$, for some $j$.
  Thus, there exists a $j_0$ such that infinitely many words of the sequence $(w_n)_{n=1}^{\infty}$
  embed to words of $D^{\B}_{j_0}$.
  We claim that $D^{\A}_i \embeds D^{\B}_{j_0}$. 
  As mentioned above, for every $w \in D^{\A}_i$, there exists an $n$ such that $w \embeds w_n$. 
  As there are infinitely many words $w_s$ embedding to $D^{\B}_{j_0}$, there exists $m \geq n$
  such that $w_m$ embeds to $D^{\B}_{j_0}$.
  Clearly, $w_n \embeds w_m$, thus $w \embeds w_n \embeds w_m$ and, hence, $w$ embeds to $D^{\B}_{j_0}$, 
  which shows that $D^{\A}_i \embeds D^{\B}_{j_0}$.
 
  Thus, there is a function $f : \{1, \ldots, k\} \to \{1, \ldots, \ell\}$ such that $D^{\A}_i \embeds D^{\B}_{f(i)}$, and 
  a function $g : \{1, \ldots, \ell\} \to \{1, \ldots, k\}$ such that $D^{\B}_j \embeds D^{\A}_{g(j)}$.
  Define the function $h : \{1, \ldots, k\} \to \{1, \ldots, k\}$ by $h(i) = g(f(i))$.
  Considering the sequence $h^1(1), h^2(1), h^3(1), \ldots$ at some moment
  we encounter a repetition since all the values come from a finite set.
  In other words there exist numbers $c, d \in \N$, $c < d$, such that $h^c(1) = h^d(1) = i$. 
  This means that $D^{\A}_i \embeds D^{\B}_{f(i)} \embeds D^{\A}_i$.
  We assign $\KA = D^{\A}_i$ and $\KB = D^{\B}_{f(i)}$.
\end{proof}

By Lemma~\ref{lem:union_free_embedded} by we can simplify the example to
\begin{align*}
  \LA_3 = a(b^* a)^* a (bb)^* ab ca bb (bc)^*
  && \text{ and } &&
  \LB_3 = b(aab)^* ba ca (b(cb^*)^* c)^* cc (cbc)^* b
\end{align*}
by eliminating $b^* c (cb)^* \subseteq \LA_2$ and $(aa)^* + ba (bb)^* \subseteq \LB_2$.

We now consider a transformation $L \mapsto L'$ transforming a language to a simpler one, 
which is still mutually embeddable with the origin.
By transitivity of the mutual-embeddability relation, 
we may transform two mutually embeddable languages and the results will also be mutually embeddable.
Next four lemmas describe this transformation. 

\begin{lemma}[Star depth one]\label{lem:star_depth_one}
  For every union-free regular language $L$, there exists a union-free regular language $L'$
  such that:
  \begin{itemize}
    \item $L'$ is of star depth at most one, i.e., it has a regular expression
    of the form $v_1 (v_2)^* v_3 \ldots (v_{2i})^* v_{2i+1}$,
    where all $v_j \in \Sigma^*$,
    \item $L$ and $L'$ are mutually embeddable,
    \item $L' \subseteq L$.
  \end{itemize}
\end{lemma}

\begin{proof}
  Consider a union-free regular expression $r$ such that $L(r) = L$.
  Then $r$ is of the form
  \[
    r = v_1 (r_2)^* v_3 (r_4)^* v_5 \ldots (r_{2i})^* v_{2i+1},
  \]
  where $v_{2j+1} \in \Sigma^*$ and $r_{2j}$ is a (union-free) regular expression, for $j \leq i$.
  Let $v_{2j}$ be obtained from $r_{2j}$ by deleting all star operations.
  For example, if $r_{2j} = a(bcd^*b)^*$ then $v_{2j} = abcdb$.
  Clearly, $L(r') \subseteq L$. Our aim is to show that $L$ and $L(r')$ with
  \[
    r' = v_1 (v_2)^* v_3 (v_4)^* v_5\ldots (v_{2i})^* v_{2i+1}
  \]
  are mutually embeddable.
  Recall from the proof of Lemma~\ref{lem:union_free_embedded} that for every word $w \in L$
  there exists $n$ such that $w \embeds w_n$ (where $w_n$ denotes the word obtained from $r$ by replacing all star operations with $n$). 
  It is now sufficient to show that $w_n \embeds L(r')$ for every $n \geq 1$.
  To this end, fix an arbitrary $n \in \N$ and assume that
  \[
    w_n = v_1 \bar{v}_2 v_3 \ldots \bar{v}_{2i} v_{2i+1},
  \]
  where $\bar{v}_{2j} \in L(r_{2j})$, for $j \leq i$.
  Note that each symbol occurring in $\bar{v}_{2j}$ also occurs in $v_{2j}$.
  Thus, $\bar{v}_{2j} \embeds v_{2j}^{|\bar{v}_{2j}|}$ and, therefore, $w_n \embeds L(r')$, which implies that $L(r) \embeds L(r')$.
  Since $L(r') \subseteq L(r)$ implies that $L(r') \embeds L(r)$, the proof is complete.
\end{proof}

Due to Lemma~\ref{lem:star_depth_one} we may eliminate the stars on depth more than one obtaining
\begin{align*}
  \LA_4 = a(ba)^* a (bb)^* ab ca bb (bc)^*
  && \text{ and } &&
  \LB_4 = b(aab)^* ba ca (bcbc)^* cc (cbc)^* b.
\end{align*}

Recall that every union-free language of star depth at most one is of the form
$
  v_1 (v_2)^* v_3 \ldots (v_{2i})^* v_{2i+1}.
$
We call the words $v_{2j}$ in the scope of a star operation \emph{loops}.
A loop $v_{2j}$ with $\Al(v_{2j}) = \Sigma_0 \subseteq \Sigma$ is called a \emph{$\Sigma_0$-loop}.

\begin{lemma}[Eliminating a loop]\label{lem:canceling_a_loop}
  Let $L$ be a regular language of the form
  $
    L = v_1 (v_2)^* v_3 \ldots (v_{2i})^* v_{2i+1}
  $
  and assume that for some $1 \leq k, \ell \leq i$, $k \neq \ell$,
  \begin{itemize}
    \item $\Al(v_{2\ell}) \subseteq \Al(v_{2k})$, and
    \item $\Al(v_j) \subseteq \Al(v_{2k})$ for all $\min(2k,2\ell) < j < \max(2k,2\ell)$.
  \end{itemize}
  Then the languages $L$ and 
  $
    L' = v_1 (v_2)^* v_3 \ldots v_{2\ell-1} v_{2\ell+1} \ldots (v_{2i})^* v_{2i+1}
  $
  obtained from $L$ by eliminating the $v_{2\ell}$ loop are mutually embeddable.
\end{lemma}

\begin{proof}
  We can assume that $k < \ell$. 
  Indeed, if the lemma holds for $k < \ell$ we can immediately infer that it also holds for $k > \ell$ 
  because $K$ and $L$ are mutually embeddable if and only if the reversed languages $K^\text{rev}$ and $L^\text{rev}$ are mutually embeddable.
  As $L' \subseteq L$ we have that $L' \embeds L$. 
  It remains to show that $L$ embeds to $L'$.
  Fix an arbitrary word $w \in L$ and assume that
  $
    w = v_1 \bar{v}_2 v_3 \ldots \bar{v}_{2i} v_{2i+1}
  $
  where $\bar{v}_{2j} \in v_{2j}^*$, for $j \leq i$.
  Note that every symbol occurring in $\bar{v}_{2k} v_{2k+1} \ldots v_{2\ell-1} \bar{v}_{2\ell}$ belongs to $\Al(v_{2k})$.
  Then $\bar{v}_{2k} v_{2k+1} \ldots v_{2\ell-1} \bar{v}_{2\ell}$ embeds to $(v_{2k})^{|w|}$, which implies that
  $
    w \embeds v_1 \bar{v}_2 \ldots v_{2k-1} (v_{2k})^{|w|} v_{2\ell+1} \ldots \bar{v}_{2i} v_{2i+1}
  $
  and, therefore, also
  $$
    w \embeds v_1 \bar{v}_2 \ldots v_{2k-1} (v_{2k})^{|w|} v_{2k+1} \ldots v_{2\ell-1} v_{2\ell+1} \ldots 
    \bar{v}_{2i} v_{2i+1} \in L',
  $$
  which completes the proof.
\end{proof}

Using Lemma~\ref{lem:canceling_a_loop} we may eliminate unnecessary loops
$(bb)^*$ in $\LA_4$ and $(bcbc)^*$ (or, alternatively, $(cbc)^*)$ in $\LB_4$ obtaining
\begin{align*}
  \LA_5 = a(ba)^* aab ca bb (bc)^*
  && \text{ and } &&
  \LB_5 = b(aab)^* ba ca cc (cbc)^* b.
\end{align*}
Note that these are the languages from Example~\ref{ex:synchronization}.

We call a union-free regular expression of star depth at most one with expressions
$v_k$ and $v_\ell$ as mentioned in Lemma~\ref{lem:canceling_a_loop} 
\emph{redundant} since, intuitively, it has a redundant loop. 
A union-free regular expression of star depth at most one that is not redundant is called \emph{nonredundant}.
We use the same notions for the corresponding languages.
In what follows, when we speak about a regular expression of a nonredundant or redundant language, 
we mean the corresponding nonredundant or redundant regular expression, respectively. 
A nonredundant regular expression of the form
\[
  v_1 (v_2)^* v_3 (v_4)^* \ldots (v_{2k})^* v_{2k+1},
\]
where $v_i \in \Sigma^*$, for $1 \leq i \leq 2k+1$, is called
\emph{saturated} if for any two loops $v_m$ and $v_n$ all symbols
from $\Al(v_m) \cup \Al(v_n)$ occur in between.
The language of a saturated expression is called \emph{saturated}.
The intuition behind this notion is explained below.
The following lemma shows that we may assume that our languages are saturated.

\begin{lemma}[Unfolding loops]\label{lem:unfolding_loops}
  Let $L$ be a nonredundant language. 
  Then there exists a saturated language $L' \subseteq L$ such that $L$ and $L'$ are mutually embeddable.
\end{lemma}

\begin{proof}
  Let the regular expression of $L$ be
  $
    r = v_1 (v_2)^* v_3 (v_4)^* \ldots (v_{2k})^* v_{2k+1}
  $
  and define
  $
    r' = v_1 v_2 (v_2)^* v_2 v_3 v_4 (v_4)^* v_4 \ldots v_{2k} (v_{2k})^* v_{2k} v_{2k+1},
  $
  where all the loops are unfolded once in every direction and the corresponding language is $L' = L(r')$.
  The nonredundancy of $L'$ is clear.
  Indeed, if there are two loops $v_i$ and $v_j$ in $r'$
  such that one of them has a bigger alphabet
  and every symbol in between $v_i$ and $v_j$ belongs to this alphabet,
  then the situation also takes place before the unfolding of the loops, in the regular expression $r$.
  Furthermore, it is easy to see that $L'$ is saturated.
  It is thus sufficient to show the mutual embeddability. Note that $L(r') \subseteq L(r)$, hence $L(r') \embeds L(r)$.
  On the other hand, every word from $L(r)$ either belongs to $L(r')$ or embeds to a word of $L(r')$
  obtained by unfolding some loops several times.
\end{proof}

After unfolding the loops $(ba)$ and $(bc)$ in $\LA_5$ and $(aab)^*$ and $(cbc)^*$ in $\LB_5$ we obtain
\[
  \LA_6 = a ba (ba)^* ba aab ca bb bc (bc)^* bc
\]
and
\[
  \LB_6 = b aab (aab)^* aab ba ca cc cbc (cbc)^* cbc b.
\]
In fact languages $\LA_5$ and $\LB_5$ were already saturated, but this is not always true in general
for nonredundant languages.

The $\Sigma_0$-decomposition of a saturated regular expression $r$ is of the form
\[
  r_1 \ u_1 (v_1)^* w_1 \ r_2 \ u_2 (v_2)^* w_2 \ldots \ r_k \ u_k (v_k)^* w_k \ r_{k+1},
\]
where 
words $v_1,v_2, \ldots, v_k$ are $\Sigma_0$-loops in $r$, 
words $u_i$ and $w_i$ satisfy $\Al(u_i)\cup \Al(w_i) \subseteq \Sigma_0$, for $1 \leq i \leq k$, and 
$r_1, r_2, \ldots, r_{k+1}$ are nonredundant expressions without $\Sigma_0$-loops
starting and ending with symbols not belonging to $\Sigma_0$.

Notice that the $\Sigma_0$-decomposition may not exist for non-saturated expressions.
Consider for instance the expression $(a b)^* a (a c)^*$, and try to compute its $\{a,b\}$-de\-com\-po\-si\-tion.
It does not exists, as, intuitively, there is no symbol outside $\{a, b\}$ between the $\{a,b\}$-loop and $\{a, c\}$-loop.
Thus it is not possible to start an expression $r_2$ by symbol not belonging to $\{a, b\}$, as required above.
This is the reason why we need to make it saturated, for example by unfolding the loops like in the proof of Lemma~\ref{lem:unfolding_loops}.
Then we obtain the expression $ab (ab)^* ab a ac (ac)^* ac$, which has the $\{a, b\}$-decomposition
of the form $r_1 = \eps$, $u_1 = ab$, $v_1 = ab$, $w_1 = abaa$ and $r_2 = c(ac)^*ac$ starting with symbol outside $\{a,b\}$, as needed.

For two saturated regular expressions $r^{\A}$ and $r^{\B}$ we say that an alphabet $\Sigma_0 \subseteq \Sigma$
is \emph{$(r^{\A}, r^{\B})$-loop-maximal} if
\begin{enumerate}
  \item there exists a $\Sigma_0$-loop either in $r^{\A}$ or in $r^{\B}$; and
  \item there is no $\Sigma' \supsetneq \Sigma_0$ for which a $\Sigma'$-loop occurs either in $r^{\A}$ or in $r^{\B}$.
\end{enumerate}
If $r^{\A}$ and $r^{\B}$ are clear from the context we simply say that an alphabet $\Sigma_0 \subseteq \Sigma$ is \emph{loop-maximal}.

\begin{lemma}[Decompositions]\label{lem:decompositions}
  Let $\LA$ and $\LB$ be two saturated and mutu\-ally-embeddable languages
  with $r^{\A}$ and $r^{\B}$ being their saturated regular expressions.
  Let $\Sigma_0 \subseteq \Sigma$ be loop-maximal. 
  Let the $\Sigma_0$-decomposition of $r^{\A}$ be
  \[
    r^{\A} = r^{\A}_1 \ u^{\A}_1 (v^{\A}_1)^* w^{\A}_1 \ r^{\A}_2 \ u^{\A}_2 (v^{\A}_2)^* w^{\A}_2 \ldots
    \ r^{\A}_k \  u^{\A}_k (v^{\A}_k)^* w^{\A}_k \ r^{\A}_{k+1}\,.
  \]
  Then the numbers of $\Sigma_0$-loops in $r^{\A}$ and $r^{\B}$ coincide.
  Moreover, the $\Sigma_0$-decompo\-si\-tion of $r^{\B}$ is
  \[
    r^{\B} = r^{\B}_1 \ u^{\B}_1 (v^{\B}_1)^* w^{\B}_1 \ r^{\B}_2 \ u^{\B}_2 (v^{\B}_2)^* w^{\B}_2 \ldots
    \ r^{\B}_k \ u^{\B}_k (v^{\B}_k)^* w^{\B}_k \ r^{\B}_{k+1}\,,
  \]
  where, for all $1 \leq i \leq k+1$, the languages $L(r^{\A}_i)$ and $L(r^{\B}_i)$ are mutually embeddable
  and saturated.
\end{lemma}

\begin{proof}
  If $v = a_1 \cdots a_k$ embeds into $w = b_1 \cdots b_l$ such that $v = b_{i_1} \cdots b_{i_k}$
  for $i_1 < \ldots < i_k$ then we say that symbol $a_j$ \emph{embeds into the position $i_j$} with respect to this embedding.
  Usually, if embedding is clear from the context, we omit it.

  We first show that both $r^{\A}$ and $r^{\B}$ have the same number of $\Sigma_0$-loops.
  For the sake of contradiction, assume that there are more $\Sigma_0$-loops in $r^{\A}$ than in $r^{\B}$.
  We will exploit the fact that $\LA \embeds \LB$. 
  Let $m$ be the size of $r^{\B}$, i.e., the number of symbols in it, and
  consider an arbitrary word
  \[
    v = s^{\A}_1 u^{\A}_1 (v^{\A}_1)^{m+1} w^{\A}_1 s^{\A}_2 u^{\A}_2 (v^{\A}_2)^{m+1} w^{\A}_2 \ldots
        s^{\A}_k u^{\A}_k (v^{\A}_k)^{m+1} w^{\A}_k s^{\A}_{k+1} \in \LA,
  \]
  where $s^{\A}_i \in L(r^{\A}_i)$, for $i \leq k+1$.
  There is a word $w \in \LB$ such that $v \embeds w$.
  Consider an arbitrary $v^{\A}_j$, for $1 \leq j \leq k$. 
  There are at least $m+1$ occurrences of $v^{\A}_j$ in $v$ and for each one the last symbol of $v^{\A}_j$ coincides with a symbol of $r^{\B}$.
  As there are $m+1$ words $v^{\A}_j$ there are also $m+1$ positions in $r^{\B}$ in which their first symbol embeds.
  By the pigeonhole principle, at least two of them coincide in $r^{\B}$.
  Recall that there is no $\Sigma'_0$-loop for $\Sigma'_0 \supsetneq \Sigma_0$ in $r^\B$.
  Thus some repeated position $x$ in $r^{\B}$ has to be inside some $\Sigma_0$-loop;
  otherwise, it would not be possible to read several words $v^{\A}_j$ and 
  after this end up in the same position in $r^{\B}$.
  Therefore we define a mapping from $\Sigma_0$-loops in $r^\A$ to $\Sigma_0$-loops in $r^\B$,
  which maps a loop from $r^\A$ to some loop in $r^\B$ in which the above discussed repeated position occurs.
  Note that there possibly could be more than one such loop in $r^\B$, then we pick one of them.

  We will show that no $\Sigma_0$-loop in $r^\B$ is assigned to two different $\Sigma_0$-loops $v^{\A}_i$ and $v^{\A}_j$ from $r^{\A}$.
  Assume, to the contrary, that both $v^{\A}_i$ and $v^{\A}_j$, for $i < j$, are mapped to the same $\Sigma_0$-loop $v^{\B}_s$ in $r^{\B}$.
  Thus every symbol in between $v^{\A}_i$ and $v^{\A}_j$ have to embed in some position in the loop $v^{\B}_s$.
  However, recall that there exists a symbol $a \notin \Sigma_0$ in $r_j^{\A}$ between the loops $v^{\A}_i$ and $v^{\A}_j$,
  while loop $v^{\B}_s$ contains only symbols from $\Sigma_0$. This leads to the contradiction.
  Therefore, in particular, there are not more $\Sigma_0$-loops in $r^{\A}$ than in $r^{\B}$.

  Thus, we may assume that the $\Sigma_0$-decomposition of $r^{\B}$ is of the form
  \[
    r^{\B} = r^{\B}_1 \ u^{\B}_1 (v^{\B}_1)^* w^{\B}_1 \ r^{\B}_2 \ u^{\B}_2 (v^{\B}_2)^* w^{\B}_2 \ldots
             \ r^{\B}_k \ u^{\B}_k (v^{\B}_k)^* w^{\B}_k \ r^{\B}_{k+1}\,.
  \]
  By definition of the $\Sigma_0$-decomposition all $r^{\A}_i$ and $r^{\B}_i$ are nonredundant.
  It remains to show that the languages $L(r^{\A}_i)$ and $L(r^{\B}_i)$ are mutually embedded.
  Fix an index $i$. We show that $L(r^{\A}_i) \embeds L(r^{\B}_i)$ since the other direction is analogous.
  Assume, to the contrary, that a word $u \in L(r^{\A}_i)$ does not embed to $L(r^{\B}_i)$.
  Note that the word $v$ above was chosen arbitrarily, with the only restriction that $\Sigma_0$-loops were repeated $m+1$ times each.
  Thus, put $s^{\A}_i = u$ and consider the position in word $w$ where the last symbol of $u$ could embed inspecting $r^\B$ from left to right.
  As shown above, $u$ cannot embed earlier than in $v^{\B}_{i-1}$. 
  Recall that the last symbol of $s^{\A}_i$ does not belong to $\Sigma_0$,
  thus it does not embed to $v^{\B}_{i-1}$ and $w^{\B}_{i-1}$.
  As $u \not\embeds L(r^{\B}_i)$, the last symbol of $u$ does not embed to the infix of $w$ corresponding to $r^{\B}_i$. 
  One more time, as the last symbol of $s^{\A}_i$ does not belong to $\Sigma_0$
  it does not embed to $u^{\B}_i (v^{\B}_i)^* w^{\B}_i$. 
  Thus, the first position where it could embed is somewhere in $r^{\B}_{i+1}$.
  Then, however, we have to assign $\Sigma_0$-loops corresponding
  to words $v^\A_{i+1}, \ldots, v^\A_k$ to $\Sigma_0$-loops corresponding
  to words $v^\B_{i+2}, \ldots, v^\B_k$ in (as shown above) an injective way, which is not possible.
\end{proof}

\begin{proof}[Proof of Lemma~\ref{lem:synchronization}]
  It is easy to see that if the languages $\KA$ and $\KB$ are nonempty and synchronized then there exists an infinite zigzag between them,
  thus also between languages $\LA$ and $\LB$.

  To prove the opposite implication, assume that there exists an infinite zigzag between the languages $\LA$ and $\LB$.
  Applying Lemma~\ref{lem:union_free_embedded} first we
  obtain nonempty union-free mutually-embeddable languages $\MA \subseteq \LA$ and $\MB \subseteq \LB$.
  Then, using Lemma~\ref{lem:star_depth_one}, several times Lemma~\ref{lem:canceling_a_loop} and, finally,
  Lemma~\ref{lem:unfolding_loops} we obtain languages $\KA$ and $\KB$ represented by saturated
  (thus also union free of star depth one and nonredundant) regular expressions
  that are mutually embeddable to the languages $\MA$ and $\MB$, respectively.
  As the mutual-embeddability relation is transitive, $\KA$ and $\KB$ are mutually embeddable.
  Note that $\KA \subseteq \MA$ and $\KB \subseteq \MB$ as the application of Lemmas~\ref{lem:star_depth_one},~\ref{lem:canceling_a_loop}~and~\ref{lem:unfolding_loops} results in sublanguages
  of the original languages. To complete the proof, we show that they are synchronized.

  Consider the regular expressions $r^{\A}$ and $r^{\B}$ (with the properties listed above) for $\KA$ and $\KB$, respectively,
  and denote the number of loops in $r^{\A}$ by $i_\A$ and in $r^{\B}$ by $i_\B$.
  We prove the rest of the lemma by induction on $i_\A + i_\B$.
  For $i_\A + i_\B = 0$, $i_1 = i_2 = 0$ and $\KA = \{w_1\}$ and $\KB = \{w_2\}$, for some $w_1, w_2 \in \Sigma^*$.
  As there exists an infinite zigzag between $\KA$ and $\KB$, we have $w_1 = w_2$ and,
  hence, $\KA$ and $\KB$ are synchronized in one step.
  Note that this is the place where we use that the languages are not necessarily disjoint.

  Assume that $i_\A + i_\B = k > 0$.   Fix an alphabet $\Sigma_0$ which is $(r^{\A}, r^{\B})$-loop-maximal.
  Then, by Lemma~\ref{lem:decompositions}, we obtain that the $\Sigma_0$-decomposition of $r^{\A}$ and $r^{\B}$ are
  \[
    r^{\A} = s^{\A}_1 \ u^{\A}_1 (v^{\A}_1)^* w^{\A}_1 \ s^{\A}_2 \ u^{\A}_2 (v^{\A}_2)^* w^{\A}_2 \ldots
              s^{\A}_k \ u^{\A}_k (v^{\A}_k)^* w^{\A}_k \ s^{\A}_{k+1}
  \]
  and
  \[
    r^{\B} = s^{\B}_1 \ u^{\B}_1\, (v^{\B}_1)^*\, w^{\B}_1 \ s^{\B}_2 \ u^{\B}_2\, (v^{\B}_2)^*\, w^{\B}_2 \ldots
              s^{\B}_k \ u^{\B}_k\, (v^{\B}_k)^*\, w^{\B}_k \ s^{\B}_{k+1}
  \]
  where the languages $L(s^{\A}_i)$ and $L(s^{\B}_i)$ are mutually embeddable and saturated for all $1 \leq i \leq k+1$.
  Thus, by induction hypothesis, all $L(s^{\A}_i)$ and $L(s^{\B}_i)$ are synchronized. 
  As, by definition, $u^{\A}_i (v^{\A}_i)^* w^{\A}_i$ and $u^{\B}_i (v^{\B}_i)^* w^{\B}_i$ are synchronized in one step, 
  for all $1 \leq i \leq k$, 
  we have that $r^{\A}$ and $r^{\B}$ are synchronized, 
  which completes the proof.
\end{proof}
 
\makeatletter{}\subsection*{Remaining Proofs of Section~\ref{sec:testing}}

\begin{proofof}{Lemma~\ref{lem:synchronizability}}
For two NFAs $\A$ and $\B$, the following conditions are equivalent.
\begin{enumerate}
  \item Automata $\A$ and $\B$ are synchronizable.
  \item There exist synchronized languages $\KA \subseteq L(\A)$ and $\KB \subseteq L(\B)$.
\end{enumerate}
\end{proofof}

\begin{proof}
The implication from left to right is immediate.
To prove the opposite implication,
let $\KA = D^{\A}_1 \ldots D^{\A}_k$ and $\KB = D^{\B}_1 \ldots D^{\B}_k$, 
where $D^{\A}_i$ and $D^{\B}_i$ are synchronized in one step, for all $1 \leq i \leq k$.
Define the \emph{$n$-th canonical word} of a singleton language as its unique word, and of a cycle language
$v_\pref (v_\midd)^* v_\suff$ as the word $v_\pref (v_\midd)^n v_\suff$.
Let $N$ be the maximum number of states of automata $\A$ and $\B$, and
let $w^{\A}_i$ and $w^{\B}_i$ be the $N$-th canonical words of languages $D^{\A}_i$ and $D^{\B}_i$, respectively, for $1 \leq i \leq k$.
Let $w^{\A} = w^{\A}_1 \cdots w^{\A}_k$ and $w^{\B} = w^{\B}_1 \cdots w^{\B}_k$.
Notice that $w^{\A} \in \KA \subseteq L(\A)$ and $w^{\B} \in \KB \subseteq L(\B)$.
Consider some of the accepting runs of $\A$ on $w^{\A}$ and of $\B$ on $w^{\B}$, respectively,
\[
  q^{\A}_1 \trans{w^{\A}_1} q^{\A}_2 \trans{w^{\A}_2} \ldots \trans{w^{\A}_{k-1}} q^{\A}_k \trans{w^{\A}_k} q^{\A}_{k+1}
\]
and
\[
  q^{\B}_1 \trans{w^{\B}_1} q^{\B}_2 \trans{w^{\B}_2} \ldots \trans{w^{\B}_{k-1}} q^{\B}_k \trans{w^{\B}_k} q^{\B}_{k+1}\,.
\]
By definition of run, states $q^{\A}_1$ and $q^{\B}_1$ are initial respectively in $\A$ and $\B$, and
states $q^{\A}_{k+1}$ and $q^{\B}_{k+1}$ are accepting in $\A$ and $\B$, respectively. 
Thus, to show that $\A$ and $\B$ are synchronizable, it is sufficient to show that 
pairs $(q^{\A}_i, q^{\B}_i)$ and $(q^{\A}_{i+1}, q^{\B}_{i+1})$ are synchronizable, for all $1 \leq i \leq k$. 
Fix some $1 \leq i \leq k$, then there are two cases. 
Either both $D^{\A}_i$ and $D^{\B}_i$ are singletons, or they are cycle languages. 
Consider first the situation when they are singletons. 
Then we have $w^{\A}_i = w^{\B}_i \in D^{\A}_i = D^{\B}_i$ and pairs $(q^{\A}_i, q^{\B}_i)$ and $(q^{\A}_{i+1}, q^{\B}_{i+1})$ are clearly synchronizable. 

Focus now on the situation where $D^{\A}_i$ and $D^{\B}_i$ are cycle languages. In this case,
\begin{align*}
  D^{\A}_i = v^{\A}_\pref (v^{\A}_\midd)^* v^{\A}_\suff
  &&\text{ and }&&
  D^{\B}_i = v^{\B}_\pref (v^{\B}_\midd)^* v^{\B}_\suff\,,
\end{align*}
for some $v^{\A}_\pref, v^{\A}_\midd, v^{\A}_\suff, v^{\B}_\pref, v^{\B}_\midd, v^{\B}_\suff \in \Sigma^*$; and
\begin{align*}
  w^{\A}_i = v^{\A}_\pref (v^{\A}_\midd)^{N} \, v^{\A}_\suff
  &&\text{ and }&&
  w^{\B}_i = v^{\B}_\pref (v^{\B}_\midd)^{N} \, v^{\B}_\suff\,.
\end{align*}
Consider a run of $\A$ on $w^{\A}_i$ from $q^{\A}_i$ to $q^{\A}_{i+1}$.
It is of the form
\[
  q^{\A}_i \trans{v^{\A}_\pref} m^{\A}_0 \trans{v^{\A}_\midd} m^{\A}_1  \trans{v^{\A}_\midd} \ldots  \trans{v^{\A}_\midd} m^{\A}_{N-1}
  \trans{v^{\A}_\midd} m^{\A}_N \trans{v^{\A}_\suff} q^{\A}_{i+1}\,,
\]
for some states $m^{\A}_j$, for $0 \leq j \leq N$.
Notice that at least two among states $m^{\A}_0, \ldots, m^{\A}_N$ necessarily coincide, as automaton $\A$ has no more than $N$ states.
Assume thus that for some $0 \leq k < \ell \leq N$ we have $m^{\A}_k = m^{\A}_\ell = m^{\A}$. Then
\[
  q^{\A}_i \trans{v^{\A}_\pref (v^{\A}_\midd)^{k}} m^{\A} \trans{(v^{\A}_\midd)^{\ell-k}} m^{\A} \trans{(v^{\A}_\midd)^{N-\ell} v^{\A}_\suff} q^{\A}_{i+1}\,,
\]
which shows that states $q^{\A}_i$ and $q^{\A}_{i+1}$ are $\Al(v^{\A}_\midd)$-connected in $\A$,
since we have $\Al(v^{\A}_\pref)\cup \Al(v^{\A}_\suff) \subseteq \Al(v^{\A}_\midd)$ by definition of languages
synchronized in one step.
Similarly we can show that $q^{\B}_i$ and $q^{\B}_{i+1}$ are $\Al(v^{\B}_\midd)$-connected in $\B$.
However, by definition of synchronization in one step, the cycle alphabets of $D^{\A}_i$ and $D^{\B}_i$ are the same,
so $\Al(v^{\A}_\midd) = \Al(v^{\A}_\midd)$. 
This shows that the pairs $(q^{\A}_i, q^{\B}_i)$ and $(q^{\A}_{i+1}, q^{\B}_{i+1})$ are synchronizable and completes the proof.
\end{proof}

\begin{proofof}{Theorem~\ref{thm:characterization}}
  Let $\A$ and $\B$ be two NFAs. 
  Then the languages $L(\A)$ and $L(\B)$ are separable by piecewise testable languages
  if and only if
  the automata $\A$ and $\B$ are not synchronizable.
\end{proofof}

\begin{proof}
  This theorem follows from the previous results.
  Namely, by Theorem~\ref{theo:general-characterization}, the languages $L(\A)$ and $L(\B)$ are separable
  by piecewise testable languages if and only if there is no infinite zigzag between them.
  Lemma~\ref{lem:synchronization} shows that the existence of a zigzag is equivalent to the existence of two synchronized
  sublanguages $\KA \subseteq L(\A)$ and $\KB \subseteq L(\B)$. Finally, by Lemma~\ref{lem:synchronizability},
  the existence of two synchronized sublanguages is equivalent to the fact that the automata $\A$ and $\B$ are synchronizable,
  which concludes the proof.
\end{proof}

\begin{proofof}{Theorem~\ref{lem:algorithm}}
  Given two NFAs $\A$ and $\B$, it is possible to test in polynomial
  time whether $L(\A)$ and $L(\B)$ can be separated by a piecewise
  testable language. 
\end{proofof}

\begin{proof}
  By Theorem~\ref{thm:characterization} it is enough to check whether
  $\A$ and $\B$ are synchronizable.  Let $\A = (Q^{\A}, \Sigma,
  \delta^{\A}, q_0^{\A}, F^{\A})$ and $\B = (Q^{\B}, \Sigma,
  \delta^{\B}, q_0^{\B}, F^{\B})$.  We will consider the graph
  $\graph$ for which the vertices are pairs of states of $Q^{\A} \times
  Q^{\B}$ and the edges correspond to pairs of vertices synchronizable in
  one step.  Specifically, there is an edge $(p^\A, p^\B) \trans{}
  (q^\A, q^\B)$ in $\graph$ if and only if $(p^\A, p^\B)$ and $(q^\A,
  q^\B)$ are synchronizable in one step.   
        Thus, $\A$ and $\B$ are synchronizable if and only if a vertex
  consisting of accepting states is reachable in $\graph$ from the
  pair of initial states $(q^{\A}_0, q^{\B}_0)$. Since reachability is
  testable in PTIME, it is thus sufficient to show how we compute the
  edges of $\graph$.

  The definition of synchronizability in one step
  (page~\pageref{def:ndfsynch}) consists of two cases.  We refer to the
  first case as \emph{symbol synchronization} and to the second case as
  \emph{cycle synchronization}.

  For two symbol-synchronizable pairs of states $(p^\A, p^\B)$,
  $(q^\A, q^\B)$, there should be an edge $(p^\A,p^\B)\to (q^\A,q^\B)$
  in $\graph$ if there exists an $a \in \Sigma$ such that $p^\A
  \trans{a} q^\A$ and $p^\B \trans{a} q^\B$.  Since it is easy to find
  all these pairs in polynomial time, these edges in $\graph$ can be
  easily constructed.

  We now show how to construct the edges for cycle-synchronized
  states.  For two pairs $(p^\A, p^\B)$ and $(q^\A, q^\B)$ to be
  cycle-synchronizable, we require that $p^\A$ and $q^\A$ are
  $\Sigma_0$-connected in $\A$, and $p^\B$ and $q^\B$ are
  $\Sigma_0$-connected in $\B$, for (the same) $\Sigma_0 \subseteq
  \Sigma$.  We now rephrase this definition using other notions that
  will be useful in the algorithm.

  A pair of states $(p^\A, p^\B) \in Q^\A \times Q^\B$ has a
  \emph{saturated $\Sigma_0$-cycle} if there exist two words $v^\A, v^\B$
  satisfying
  \begin{enumerate}
  \item $\Al(v^\A) = \Al(v^\B) = \Sigma_0$;
  \item $p^\A \trans{v^\A} p^\A$ in $\A$; and
  \item $p^\B \trans{v^\B} p^\B$ in $\B$.
  \end{enumerate}
  We say that there is a \emph{$\Sigma_0$-route} from $(p^\A,
  p^\B) \in Q^\A \times Q^\B$ to $(q^\A, q^\B) \in Q^\A \times Q^\B$
  if there exist words $v^\A$ and $v^\B$ in $\Sigma_0^*$ such that
  $p^\A \trans{v^\A} q^\A$ and $p^\B \trans{v^\B} q^\B$.  So, in
  contrast to saturated $\Sigma_0$-cycles, here we do not require that the
  alphabets $\Al(v^\A)$ and $\Al(v^\B)$ are equal to $\Sigma_0$.

  Note that if a pair $V = (q^\A,q^\B)$ has a saturated
  $\Sigma_0$-cycle and a saturated $\Sigma_1$-cycle, then it also has
  a saturated $(\Sigma_0 \cup \Sigma_1)$-cycle (obtained by the
  concatenation of the two cycles).  Thus, for every pair $V =
  (q^\A,q^\B)$, there exists a unique maximal alphabet $\Sigma_0
  \subseteq \Sigma$ such that it has a saturated
  $\Sigma_0$-cycle. (This unique maximal alphabet can be empty if no
  such saturated cycle exists.)  We call this alphabet the
  \emph{saturated cycle alphabet} of $V$ and denote it by
  $\Sigma_0^V$.  This means that $V_p = (p^\A, p^\B)$ and $V_q =
  (q^\A, q^\B)$ are cycle synchronizable if and only if there is a $V$
  such that there are $\Sigma_0^V$-routes from $V_p$ to $V$ and
  from $V$ to $V_q$.

  To find all the cycle synchronizable pairs we can first compute, for
  every $V = (p^\A,p^\B)$, the saturated cycle alphabet
  $\Sigma_0^V$.  This can be done in polynomial time in the following
  manner.  Let $C^{\A}_0$ and $C^{\B}_0$ be the strongly connected
  components of $\A$ and $\B$ containing $p^\A$ and $p^\B$,
  respectively.  For a strongly connected component $C$, let $\Al(C)$ be the
  union of all symbols $a$ of $\Sigma$ that label transitions of the
  form $p \trans{a} q$, where both $p$ and $q$ belong to $C$.  If
  $\Al(C^\A_0) = \Al(C^\B_0)$, then $\Sigma_0^V$
  equals $\Al(C^\A_0)$.  Otherwise, set $\Sigma_1 = \Al(C^\A_0) \,
  \cap \, \Al(C^\B_0)$ and consider automata $\A_1$ and $\B_1$
  obtained from $\A$ and $\B$ by removing all transitions labeled by
  symbols from $\Sigma\setminus\Sigma_1$.  Consider the strongly
  connected components $C^\A_1$ and $C^\B_1$ of $\A_1$ and $\B_1$
  containing $p^\A$ and $p^\B$, respectively, and proceed in the same
  way as before.  Continuing this procedure we obtain a sequence of
  decreasing alphabets $\Sigma_1 \supsetneq \Sigma_2 \supsetneq
  \ldots$, hence we perform at most $|\Sigma|$ iterations. If we
  arrive at the empty alphabet then we say $\Sigma_0^V = \emptyset$.

  We argue that we compute $\Sigma_0^V$ correctly. Clearly, if the
  algorithm returns a set $\Sigma'$, then $\Sigma' \subseteq
  \Sigma_0^V$. Conversely, we have that $\Sigma_0^V \subseteq \Sigma'$
  because, at each point in the algorithm, the alphabet under
  consideration contains $\Sigma_0^V$. (In the first iteration,
  $\Sigma_0^V \subseteq \Al(C^\A_0)$ and $\Sigma_0^V \subseteq
  \Al(C^\B_0)$. Furthermore, at each iteration $i$, if $\Sigma_0^V
  \subseteq \Al(C^\A_i)$ and $\Sigma_0^V \subseteq \Al(C^\B_i)$, then
  $\Sigma_0^V \subseteq (\Al(C^\A_i) \cap \Al(C^\B_i))$.)

  Once we know, for each pair $V = (q^\A,q^\B)$, its saturated cycle alphabet
  $\Sigma_0^V$, we can find all vertices $V_p$ such that
  there is a $\Sigma_0^V$-route from $V_p$ to $V$ and all vertices
  $V_q$ such that there is a $\Sigma_0^V$-route from $V$ to $V_q$,
  and add edges $V_p\to V_q$ to the graph $\graph$. This concludes the
  construction of $\graph$ and the presentation of the algorithm.
  We note that our algorithm is clearly not yet time-optimal.
                  \end{proof}

\makeatletter{}\section*{Proofs of Section~\ref{sec:others}}

The goal is to prove the following Theorem.

\begin{proofof}{Theorem~\ref{theo:other}}
  For $O \in \{\preceq,\sufforder\}$ and $C$ being one of single,
  unions, or boolean combinations, we have that the complexity of the
  separation problem by $\F(O,C)$ is as indicated in
  Table~\ref{tab:overview}.
\end{proofof}

\subsection*{The Subsequence Order Cases}

\begin{lemma}\label{lem:sepNPc}
  The separation problem by $\F(\preceq,\text{single})$ is NP-complete.
\end{lemma}

\begin{proof}
  Let $K$ and $L$ be two regular languages over $\Sigma$ given by
  NFAs.  The problem is to find a word $w$ in $\Sigma^*$ such that $K
  \subseteq \closure^\preceq(w)$ and $L \cap \closure^\preceq(w) =
  \emptyset$.  By definition of the subsequence order $\preceq$, the
  maximal length of such a word $w$ is equal to the length of a
  shortest word of $K$. Therefore, such a $w$ cannot be longer than
  the size of the automaton for $K$.  An NP algorithm can guess such a
  word $w$ of length at most the size of the automaton and computes
  the minimal DFA for $\closure^\preceq(w)$. This minimal DFA
  corresponds to a ``greedy'' procedure for embedding $w$ in a given
  string. That is, the states of this DFA correspond to the maximal
  prefix of $w$ that can be embedded in the currently read string. It
  can be computed in polynomial time from $w$.
  Verifying if $K \subseteq \closure^\preceq(w)$ and $L \cap
  \closure^\preceq(w) = \emptyset$ then reduces to standard automata
  constructions that can be done in polynomial time.

  To show NP-hardness, we use a simple reduction of the longest common
  subsequence problem, which is well known to be
  NP-hard~\cite{DBLP:journals/jacm/Maier78}.  A word $w$ is a
  \emph{longest common subsequence} of words $(w_i)_{i=1}^n$ if $w
  \preceq w_i$ for all $1 \leq i \leq n$ and there is no longer word
  with this property.  This word $w$ is not necessarily unique (the
  longest common subsequence for $ab$ and $ba$ could be $a$ or $b$).
  By~\cite{DBLP:journals/jacm/Maier78}, to determine whether the
  length of the longest common subsequences of words $(w_i)_{i=1}^n$
  is longer than a given $k$ is NP-hard with respect to $\sum_{i=1}^n
  |w_i|$ and $k$.  

  Consider the DFA $\A$ that accepts the finite language $K=\{w_1,
  \ldots, w_n\}$ and the DFA $\B$ that accepts the language $L$ of all
  words up to length $k-1$. Then we have that the existence of a
  common subsequence of $(w_i)_{i=1}^n$ longer than $k$ is then
  equivalent to the possibility to separate $K$ and $L$ by
  $\F(\preceq,\text{single})$. Furthermore, we can construct $\A$ in
  time $O(\sum_{i=1}^n w_i)$ and $\B$ in time $O(k \cdot \sum_{i=1}^n
  w_i)$. Since both $\A$ and $\B$ are DFAs, we have shown that the
  problem even remains NP-hard if the input is given as DFAs instead
  of NFAs.   
\end{proof}

Actually, using the proof of Lemma~\ref{lem:sepNPc} we can prove the
same result for union-free languages.

\begin{lemma}\label{lem:union-free}
  The separation problem by union-free languages is NP-complete.
\end{lemma}

\begin{proof}
  The proof of NP-hardness of the proof of Lemma~\ref{lem:sepNPc} also
  applies to union-free languages since the language
  $\closure^\preceq(a_1a_2\ldots
  a_n)=\Sigma^*a_1\Sigma^*a_2\Sigma^*\ldots\Sigma^*a_n\Sigma^*$ is
  union free.  Indeed, any regular expression
  $(b_1+b_2+\ldots+b_m)^*=(b_1^*\ldots b_m^*)^*$ is union free.  The
  NP algorithm guesses a word $w$ as above and the positions and
  scopes of star operators. 
\end{proof}

We now turn to separation by $\F(\preceq,\text{unions})$.

\begin{lemma}\label{lem:union}
  A language $K$ is separable from a language $L$ by 
  $\F(\preceq,\text{unions})$ if and only if there exist no words $w \in K$
  and $w' \in L$ such that $w \preceq w'$.
\end{lemma}
\begin{proof}
  If there exist $w \in K$ and $w' \in L$ with $w\preceq w'$, then any
  $\preceq$-closed language containing $w$ also contains $w'$. Since
  unions of $\preceq$-closed languages are also $\preceq$-closed, we
  have that $K$ is not separable from $L$ by $\F^{\preceq}_\union$.

  The opposite implication follows directly from
  Claim~\ref{claim:one_layer}. Observe that in this case $\layer(K, L)
  = K$ and every $\preceq$-closed language is a finite union of
  languages $\Sigma^* a_1 \Sigma^* \cdots \Sigma^* a_n \Sigma^*$ due
  to Higman's lemma.  
\end{proof}

The words $w$ and $w'$ from the statement of Lemma~\ref{lem:union} exist iff
$\closure^\preceq(K) \cap L = \emptyset$.  An NFA for $\closure^\preceq(K)$ is
obtained by adding self loops under all symbols in $\Sigma$ to all
states of the automaton for $K$.  Emptiness of intersection is then
decidable in polynomial time by standard methods. This gives the
following lemma.

\begin{lemma}\label{lem:union-ptime}
  The separation problem by $\F(\preceq,\text{unions})$ is in polynomial time.
\end{lemma}

\subsection*{The Suffix Order Cases}

It remains to prove the cases for the suffix order
$\sufforder$.  Let $\lcs(L)$ denote the
longest common suffix of all words of language $L$.

\begin{lemma}\label{lem:sufforder}
  A language $K$ is separable from a language $L$ by
  $\F(\sufforder,\text{single})$ if and only if there is no word $w' \in L$ such
  that $\lcs(K) \, \sufforder \, w'$.
\end{lemma}
\begin{proof}
  The separation problem asks to check the existence of a word $w \in \Sigma^+$
  such that $K \subseteq \Sigma^* w$ and $\Sigma^* w \cap L =
  \emptyset$.  Obviously, if such a word exits, it must be a common
  suffix of all words from $K$.
  Assume that there is no $w' \in L$ such that $\lcs(K) \, \sufforder
  \, w'$.  Then $K$ is separable from $L$ by the language $\Sigma^*
  \lcs(K)$.  To show the opposite implication, assume that there
  exists a $w'\in L$ such that $\lcs(K) \, \sufforder \, w'$.  Then, for
  any common suffix $w$ of $K$, it holds that $w \, \sufforder \, w'$,
  which means that $K$ is not separable from $L$ by a language from
  $\F(\sufforder,\text{single})$.  
\end{proof}

The word $\lcs(K)$ can be computed from the automaton for $K$ in
polynomial time by inspecting paths that end up in accepting states.  The
length of $\lcs(K)$ is not larger than the length of the shortest word
in $K$, hence linear with respect to the size of the automaton.  To
check whether there exists $w' \in L$ such that $\lcs(K) \, \sufforder
\, w'$ can be done in polynomial time by testing non-emptiness of the
language $\Sigma^*\lcs(K) \cap L$.

\begin{lemma}
  The separation problem by $\F(\sufforder,\text{single})$ is in
  polynomial time.
\end{lemma}

\begin{lemma}\label{lem:suffixunion}
  A language $K$ is separable from a language $L$ by
  $\F(\sufforder,\text{unions})$ if and only if the following two conditions
  are satisfied:
  \begin{enumerate}
  \item there exist no words $w \in K$ and $w' \in L$ such that $w \,
    \sufforder \, w'$,
  \item there exists a natural number $k\geq 0$ such that no words $w
    \in K$ and $w' \in L$ have a common suffix of length $k$.
  \end{enumerate}
\end{lemma}
\begin{proof}
  From left to right.  Assume that $K$ is separable from $L$ by a
  language $S = \bigcup_{i=1}^{n} \Sigma^* w_i$. If $w \in K$ and $w
  \in S$ then there is no $w' \in L$ such that $w \, \sufforder \, w'$
  since $S$ contains all words that are longer than $w$ in
  $\sufforder$ and $S \cap L = \emptyset$.  Assume that for every
  number $k$ there are words $w \in K$ and $w' \in L$ with a common
  suffix of length $k$.  Then, in particular, there are words $w \in
  K$ and $w' \in L$ with a common suffix of length $\max(|w_1|,
  \ldots, |w_n|) + 1$.  However, these words are either both inside
  $S$ or both outside $S$, which contradicts that $K$ is separable
  from $L$ by $S$.  This concludes the proof from left to right.

  For the other direction, assume that $K$ and $L$ satisfy conditions
  1 and 2. Let $ M = \{w\in\Sigma^* \mid |w| \leq k \text{ and there
    is no } w' \in L \text{ such that } w \, \sufforder \, w' \} $ and
  define $ S = \bigcup_{w \in M} \Sigma^* w.  $ By definition, $S \cap
  L = \emptyset$ and $S$ is a finite union of suffix languages, i.e.,
  a finite union of $\sufforder$-closures of words. We show that $K
  \subseteq S$. Indeed, let $w \in K$. If $|w| \geq k$ and $v$ is a
  suffix of $w$ of length $k$ then $v$ belongs to $M$, which implies
  that $w \in \Sigma^* v \subseteq S$.  If $|w| < k$ then $w \in M$
  since there is no $w' \in L$ such that $w \, \sufforder \, w'$.
  Thus, $w \in S$, which completes the proof.  
\end{proof}

We now argue that the two conditions in Lemma~\ref{lem:suffixunion}
can be tested in polynomial time, given NFAs for $K$ and $L$.  To
check the first condition we test in polynomial time whether
$(\Sigma^* K) \cap L$ is nonempty.  To decide the second condition we
compute the reversals $\rev(K)$ and $\rev(L)$ of languages $K$ and
$L$, respectively.  This is done by reversing transitions in the
corresponding automata and swapping the role of initial and accepting
states. We note that this step may require an NFA to have more than
one initial state, but NFAs are known to be sufficiently robust to
allow this.  Common suffixes of words from $K$ and $L$ are common
prefixes of words from $\rev(K)$ and $\rev(L)$.  We then compute the
language of all prefixes of words from $\rev(K)$ and $\rev(L)$ by
making all the states accepting, thereby obtaining languages
$\text{pref}(\rev(K))$ and $\text{pref}(\rev(L))$, respectively.  The
intersection $ I = \text{pref}(\rev(K)) \cap \text{pref}(\rev(L)) $ is
the set of all words $v \in \Sigma^*$ such that there are words $w \in
K$ and $w' \in L$ with $v \, \sufforder \, w$ and $v \, \sufforder \,
w'$.  To check the condition it is sufficient to test whether the
language $I$ is infinite, which can also be done in polynomial
time. This leads to the following lemma.
\begin{lemma}\label{lem:suffixunion-ptime}
  The separation problem by $\F(\sufforder,\text{union})$ is in
  polynomial time.
\end{lemma}

\begin{lemma}\label{lem:suffixbool}
  A language $K$ is separable from a language $L$ by
  $\F(\sufforder,\text{bc})$ if and only if there exists a natural
  number $k\geq 0$ such that no words $w \in K$ and $w' \in L$ have a
  common suffix of length $k$.
\end{lemma}

\begin{proof}
  Assume that $K$ is separable from $L$ by a finite boolean
  combination of languages $\Sigma^* w_1, \ldots, \Sigma^* w_n$.  Let
  $k = \max(|w_1|, \ldots, |w_n|) + 1$. Note that, for all words $w v$
  and $w' v$ with $|v| \geq k$ and all $1 \leq i \leq n$, it holds
  that $ w v \in \Sigma^* w_i \text{ if and only if } w' v \in
  \Sigma^* w_i.  $ Thus, any words with a common suffix of length at
  least $k$ cannot be separated by the considered set of languages,
  which means that there are no words $w \in K$ and $w' \in L$ with a
  common suffix of length $k$.

  To show the opposite implication, assume that there exists a natural
  number $k$ satisfying the condition.  Then, for every $w\in K$, if
  $|w|<k$, we can cover word $w$ by the language $\{w\}=\Sigma^*w
  \setminus \bigcup_{a\in \Sigma} \Sigma^* aw$.  If $|w|\ge k$, then
  $w\in\Sigma^* v$, where $v$ is a suffix of $w$ of length $k$.  By the
  assumption that no words of $K$ and $L$ have a common suffix of
  length $k$ we have that $\Sigma^* v\cap L=\emptyset$, which
  completes the proof. 
\end{proof}

It therefore follows that separability by $\F(\sufforder,\text{bc})$
can be done with a simplified version of the procedure for
$\F(\sufforder,\text{unions})$. Since the latter was already in
polynomial time according to Lemma~\ref {lem:suffixunion-ptime}, we have the following lemma.

\begin{lemma}\label{lem:suffixbool-ptime}
  The separation problem by $\F(\sufforder,\text{bc})$ is in
  polynomial time.
\end{lemma}

\end{document}